\documentclass[runningheads]{llncs}\synctex=1

\usepackage{color}
\usepackage[inline]{enumitem}%
\usepackage{stmaryrd}

\usepackage{thmtools,thm-restate}
\usepackage{nicefrac}
\usepackage{proof}

\usepackage{mathtools}

\usepackage{xifthen}
\usepackage{ifdraft}
\usepackage[all]{xy}

\usepackage{relsize}

\usepackage{wrapfig}

\usepackage{xspace}

\ifdefined\needsPoormanCleveref%
  \usepackage[poorman]{cleveref}%
\else
  \usepackage{cleveref}%
\fi%

\usepackage{amsmath}
\usepackage{amssymb}

\newcommand{\T}{{\mathsf T}}%

\newcommand{\ST}{{\mathsf S}}

\newcommand{\rulename}[1]{\text{\scriptsize[\textsc{#1}]}}

\newcommand{\pp}{{\sf p}}
\newcommand{\q}{\pq}
\newcommand{\pq}{{\sf q}}
\newcommand{\pr}{{\sf r}}

\newcommand{\x}{x}
\newcommand{\y}{y}
\newcommand{\z}{z}
\newcommand{\h}{h}
\newcommand{\val}{\kf{v}}

\newcommand{\sep}{\ensuremath{~~\mathbf{|\!\!|}~~ }}
\newcommand{\kf}[1]{\ensuremath{\mathsf{#1}}}
\newcommand{\pc}{\ensuremath{~|~}}

\newcommand{\ty}{\textbf{t}}

\newcommand{\N}{\ensuremath{\mathcal M}}
\newcommand{\M}{\ensuremath{\mathcal M}}

\newcommand{\pa}[2]{#1 \triangleleft  #2}
\newcommand{\set}[1]{\{#1\}}
\newcommand{\true}{\kf{true}}
\newcommand{\false}{\kf{false}}

\newcommand{\sub}[2]{\set{#1/#2}}

\newcommand{\fend}{\mathtt{end}}
\newcommand{\subst}[2]{\{\nicefrac{#1}{#2}\}}%

\newcommand{\msgLabel}[1]{\mathit{{#1}}}%
\newcommand{\msgPayload}[1]{\mathsf{{#1}}}%
\newcommand{\procin}[3]{#1 ? #2.#3}%
\newcommand{\procinNoSuf}[2]{{#1} ? {#2}}

\newcommand{\procinconcur}[3]{\|#1\|_{#2}.#3}

\newcommand{\procinconcurNoSuf}[1]{\|#1\|}

\newcommand{\procout}[4]{#1 ! #2 \langle #3 \rangle . #4 }%
\newcommand{\procoutNoSuf}[3]{{#1} ! {#2}\langle{#3}\rangle}

\newcommand{\PP}{\ensuremath{P}}
\newcommand{\Q}{\PQ}
\newcommand{\PQ}{\ensuremath{Q}}
\newcommand{\PR}{\ensuremath{R}}
\newcommand{\PRio}{\ensuremath{R^{\mathit{?!}}}}
\newcommand{\cond}[3]{\kf{if}~ #1 ~\kf{then} ~#2 ~\kf{else}~#3}
\newcommand{\inact}{\mathbf{0}}

\newcommand{\emptyqueue}{\varnothing}

\newcommand{\tend}{\mathtt{end}}

\newcommand{\tbool}{\mathtt{bool}}
\newcommand{\tnat}{\mathtt{nat}}

\newcommand{\tin}[3]{#1?#2(#3)}
\newcommand{\tout}[3]{#1!#2 \langle #3 \rangle}

\DeclareMathOperator*{\tinternal}{%
  \vphantom{\sum}\mathlarger{\mathlarger{\mathlarger{\mathlarger{\raisebox{-3pt}{$\Sigma$}}}}}%
}%
\DeclareMathOperator*{\texternal}{%
  \vphantom{\sum}\mathlarger{\mathlarger{\mathlarger{\mathlarger{\raisebox{-3pt}{$\Sigma$}}}}}%
}%

\newcommand{\tqueue}{\sigma}
\newcommand{\temptyqueue}{\epsilon}

\renewcommand{\S}{\ST}%

\newcommand{\subt}{\leqslant}

\newcommand{\red}{\longrightarrow}
\newcommand{\recvLabel}[3]{{#1}{:}{#2}?{#3}}
\newcommand{\sendLabel}[3]{{#1}{:}{#2}!{#3}}
\newcommand{\ifLabel}[1]{{#1}{:}\kf{if}}
\newcommand{\redLabel}[1]{\xrightarrow{#1}}
\newcommand{\redSend}[3]{\xrightarrow{\sendLabel{#1}{#2}{#3}}}
\newcommand{\redRecv}[3]{\xrightarrow{\recvLabel{#1}{#2}{#3}}}
\newcommand{\redIf}[1]{\xrightarrow{\ifLabel{#1}}}
\newcommand{\reds}{\mathrel{\longrightarrow^{\!*}}}
\newcommand{\redPlus}{\mathrel{\longrightarrow^{\!+}}}
\newcommand{\nred}{{\,\,\not\!\!\longrightarrow}}%
\newcommand{\equivv}{\Rrightarrow}

\newcommand{\cinferrule}[3][]{
  \infer=[#1]{#3}{#2}%
}
\newcommand{\inferrule}[3][]{
  \infer[\!\!{#1}]{#3}{#2}%
}

\newcommand{\inferruleR}[3][]{
  \infer[\!\!{#1}]{#3}{#2}%
}

\newcommand{\cinfer}[3][]{
  \infer=[#1]{#3}{#2}%
}%

\newcommand{\dom}[1]{\ensuremath{dom( #1)}}

\newcommand{\EmptyQueue}{\varnothing}
\newcommand{\Queue}{h}
\newcommand{\msg}[3]{(#1,#2(#3))}

\usepackage[T1]{fontenc}
\usepackage{graphicx}
\begin{document}
\title{On Asynchronous Multiparty Session Types for Federated Learning}
\author{Ivan Proki\'c\inst{1}\orcidID{0000-0001-5420-1527} \and
Simona Proki\'c\inst{1}\orcidID{0000-0002-7161-3926} \and
Silvia Ghilezan\inst{1,2}\orcidID{0000-0003-2253-8285} \and
Alceste Scalas\inst{3}\orcidID{0000-0002-1153-6164} \and
Nobuko Yoshida\inst{4}\orcidID{0000-0002-3925-8557}}
\authorrunning{Ivan Proki\' c et al.}

\institute{Faculty of Technical Sciences, University of Novi Sad, Serbia \and
 Mathematical Institute of the Serbian Academy of Sciences and Arts, Belgrade, Serbia, 
 \email{\{prokic,simona.k,gsilvia\}@uns.ac.rs}\\
 \and
 Technical University of Denmark, DK,  
\email{alcsc@dtu.dk}\\ \and
University of Oxford, UK,  
\email{nobuko.yoshida@cs.ox.ac.uk}}
\maketitle              
\begin{abstract}
This paper improves the session typing theory to support the modelling and verification of processes that implement federated learning protocols. 
To this end, we build upon the asynchronous ``bottom-up'' session typing approach by adding support for input/output operations directed towards multiple participants at the same time. 
We further enhance the flexibility of our typing discipline and allow for safe process replacements by introducing a session subtyping relation tailored for this setting. 
We formally prove safety, deadlock-freedom, liveness, and session fidelity properties for our session typing system. %
Moreover, we highlight the nuances of our session typing system, which (compared to previous work) reveals interesting interplays and trade-offs between safety, liveness, and the flexibility of the subtyping relation.

\keywords{Multiparty session types  \and Federated learning \and $\pi$-calculus \and Type systems.}
\end{abstract}

\section{Introduction}\label{sec:intro}
Asynchronous multiparty session types \cite{HYC16} provide a formal approach to the verification of message-passing programs. 
The key idea is to express \emph{protocols as types}, and use type checking to verify whether one or more communicating processes correctly implement some desired protocols. 
To enhance usability of session type theory in real-world applications, many extensions and variations of the approach have been proposed over the years \cite{HuY17,ScalasY19,GheriLSTY22,MajumdarMSZ21,Stutz23,LiSWZ23,StutzD25}.   
However, these extensions remain insufficient for several important applications, including Federated learning (FL). 
FL is a distributed machine learning setting where clients train a model while keeping the training data decentralized \cite{McMahanMRHA17,BhuyanM22,BeltranPSBBPPC23,YuanWSYB24,popovic2023,ProkicGKPPK23}.  
In FL, communication protocols follow key patterns that must be expressed for proper modeling and verification. 
For instance, some initial investigations \cite{popovic2023,ProkicGKPPK23} note that using  existing session type theories to model and verify FL protocols can be challenging due to presence of ``arbitrary order of message arrivals''. We now explain the nature of these challenges. 

\emph{Modelling an asynchronous centralised federated learning protocol.} 
Originally, FL considered centralised approach \cite{McMahanMRHA17,BhuyanM22,popovic2023,ProkicGKPPK23},  where, in \emph{phase 1}, a central \emph{server} distributes the learning model to the \emph{clients}. In \emph{phase 2}, the clients receive the model, train it locally, and send the updated version back to the server. Finally, the server aggregates the updates and obtains an improved model. 
As a concrete example, consider a single round of a \emph{generic centralised one-shot federated learning algorithm (FLA)} \cite{popovic2023,ProkicGKPPK23} having one server and two clients. 
We may model the processes for server $\pp$ and two clients $\pq$ and $\pr$ with a value passing labeled $\pi$-calculus as follows
--- where $\msgLabel{ld}$ and $\msgLabel{upd}$ are message labels meaning ``local data'' and ``update,'' $\sum$ represents a choice, and %
and $\procoutNoSuf{\pq}{\msgLabel{ld}}{\msgPayload{\ldots}}$ (resp.~$\procinNoSuf{\pq}{\msgLabel{ld}(\ldots)}$) means ``send (resp.~receive) the message $\msgLabel{ld}$ to (resp.~from) participant $\pq$.''
\begin{align*}
\label{intro:fla_centr}
\PP & = \procoutNoSuf{\pq}{\msgLabel{ld}}{\msgPayload{data}}\,.\,
    \procoutNoSuf{\pr}{\msgLabel{ld}}{\msgPayload{data}}. 
    	\sum\biggl\{%
\begin{array}{@{}l@{}}
    \procinNoSuf{\pq}{\msgLabel{upd}(\x)}\,.\,
    \procinNoSuf{\pr}{\msgLabel{upd}(\y)}
 \\
    \procinNoSuf{\pr}{\msgLabel{upd}(\y)}\,.\,
    \procinNoSuf{\pq}{\msgLabel{upd}(\x)}
\end{array}
\biggr\} \\
\PQ & = \PR  =
\procinNoSuf{\pp}{\msgLabel{ld}(x)}.\procoutNoSuf{\pp}{\msgLabel{upd}}{\msgPayload{data}}
\end{align*}
In phase 1, $\pp$'s process $\PP$ above sends its local data (i.e., a machine learning model) to $\pq$ and then to $\pr$.  
In phase 2, server $\pp$ receives the update $\msgLabel{upd}$ from $\pq$
and also $\pr$; this can happen in two possible orders depending on whether the data is received from $\pq$ or $\pr$ first (i.e., arbitrary order of message arrivals),
and this is represented by the two branches of the choice $\sum$. Clients $\pq$ and $\pr$ receive server's local data and reply with update. 

Furthermore, suppose we have implemented the above process specifications and verified that they behave as intended. Now, we want to upgrade the participant $\pq$ to support for \emph{multi-model FL} \cite{BhuyanM22}, in which  a client can train multiple models, but only one per round (due to computational constraints). Assuming $\pq$ can train two models $\msgLabel{ld}$ and $\msgLabel{ld'}$, we may model an updated process $\PQ'$. 
\label{intro:process_update_for_multi_model_fl}
\begin{align*}
\PQ' = 
	\sum\biggl\{%
\begin{array}{@{}l@{}}
\procinNoSuf{\pp}{\msgLabel{ld}(x)}.\procoutNoSuf{\pp}{\msgLabel{upd}}{\msgPayload{data}} \\
\procinNoSuf{\pp}{\msgLabel{ld'}(y)}.\procoutNoSuf{\pp}{\msgLabel{upd'}}{\msgPayload{data'}}
\end{array}
\biggr\}
\end{align*}
This rises a question: \emph{Can we safely substitute $\PQ$ with $\PQ'$ in the protocol without needing to re-verify the entire protocol?} 

\emph{Modelling an asynchronous decentralised federated learning protocol.} 
Now let us consider a class of asynchronous decentralised federated learning (DFL) protocols that rely on a fully connected network topology (in which direct links exist between all pairs of nodes) 
\cite{BeltranPSBBPPC23,YuanWSYB24,popovic2023,ProkicGKPPK23}. 
Concretely, we consider a single round of a \emph{generic decentralised one-shot federated learning algorithm (FLA)} \cite{popovic2023,ProkicGKPPK23} having three nodes, 
where there is no central point of control (we provide a message sequence chart in \Cref{app:intro}). 
Each participant in this algorithm follows the same behaviour divided in three phases. 
In \emph{phase 1}, each participant sends its local 
data (i.e., a machine learning model) to all other participants.  
In \emph{phase 2}, each participant receives the other participants' local data, trains the algorithm, and sends back the updated data.
Finally, in \emph{phase 3}, each participant receives the updated data from other participants and aggregates the updates. 

Again, we may model the process for participant $\pp$ with a value passing labeled $\pi$-calculus. 
\label{intro:fla}%
\begin{align*}
\PP_{12} & = \procoutNoSuf{\pq}{\msgLabel{ld}}{\msgPayload{data}}\,.\,
    \procoutNoSuf{\pr}{\msgLabel{ld}}{\msgPayload{data}}.
\sum\biggl\{%
\begin{array}{@{}l@{}}
    \procinNoSuf{\pq}{\msgLabel{ld}(\x)}\,.\,
    \procoutNoSuf{\pq}{\msgLabel{upd}}{\msgPayload{data}}\,.\,
    \procinNoSuf{\pr}{\msgLabel{ld}(\y)}\,.\,
    \procoutNoSuf{\pr}{\msgLabel{upd}}{\msgPayload{data}}.
    \PP_3
 \\
    \procinNoSuf{\pr}{\msgLabel{ld}(\y)}\,.\,
    \procoutNoSuf{\pr}{\msgLabel{upd}}{\msgPayload{data}}\,.\,
    \procinNoSuf{\pq}{\msgLabel{ld}(\x)}\,.\,
    \procoutNoSuf{\pq}{\msgLabel{upd}}{\msgPayload{data}}.
    \PP_3
\end{array}
\biggr\} \\
\PP_3 & = \sum\biggl\{%
    \begin{array}{@{}l@{}}
        \procinNoSuf{\pq}{\msgLabel{upd}(\x)}\,.\,
        \procinNoSuf{\pr}{\msgLabel{upd}(\y)}
    \\%
        \procinNoSuf{\pr}{\msgLabel{upd}(\y)}\,.\,
        \procinNoSuf{\pq}{\msgLabel{upd}(\x)}
    \end{array}
    \biggr\}
\end{align*}
In phase 1, $\pp$'s process $\PP_{12}$ above sends its local data to $\pq$ and then to $\pr$.  
In phase 2 (starting with the sum), participant $\pp$ receives local data and replies the update $\msgLabel{upd}$ to $\pq$ and also $\pr$; this can happen in two possible orders depending on whether the data is received from $\pq$ or $\pr$ first. 
Finally, in phase 3 (in process $\PP_3$), $\pp$ receives update data from $\pq$ and $\pr$ in any order (again, this is represented with two choice branches).

Notice that each participant here exhibits an arbitrary order of message arrivals, and that no single participant guides the protocol. This can  be challenging to model using session type theories that rely on global types.

\emph{Contributions and outline of the paper.}
We present a novel ``bottom-up'' asynchronous session typing theory (in the style of  \cite{ScalasY19}, i.e., that does not require global types) that supports input/output operations directed towards multiple participants. This enables the modeling of arbitrary  message arrival orders. %
We demonstrate that our approach enables the modelling and enhances verification of processes implementing  
federated learning protocols (both centralized and decentralized) %
by abstracting protocol behavior to the level of types. 
We enhance flexibility of our typing discipline and allow for safe process replacements  
by introducing a session subtyping relation tailored for this setting. 
Furthermore, we formalise and prove safety, deadlock-freedom, liveness, and session fidelity properties for well-typed processes, revealing interesting dependencies between these properties in the presence of a subtyping relation.

The paper is organized as follows: 
\Cref{sec:msc} introduces the asynchronous session calculus; 
\Cref{sec:tsu} presents the types and subtyping relation; \Cref{sec:tsy} details the type system, our main results, and the implementation of federated learning protocols; and finally, \Cref{sec:related} concludes with a discussion of related and future work.

\section{The Calculus}\label{sec:msc}
In this section we present the syntax 
and operational semantics 
of the value-passing labeled $\pi$-calculus used in this work (we omit session creation and  delegation).
\begin{figure}[t]
\(
  \begin{array}{@{}rcll@{}}
  \val & \Coloneqq & 
    \x, \y, \z, \ldots 
    \sep 1, 2, \ldots  
    \sep \true, \false & \text{\em (variables, values)} \\[1mm]
    \N & \Coloneqq &  
       \pa\pp\PP \pc \pa\pp \h 
       \sep \N\pc\N  & \text{\em (participant, parallel)}\\[1mm]
    \h & \Coloneqq & 
       \emptyqueue 
       \sep \left(\q,\ell(\val)\right) 
       \sep \h\cdot \h  &\text{\em (empty, message, concatenation)}\\[1mm]
    \PP,\PQ & \Coloneqq & 
      		\sum_{i\in I}\procin{\pp_i}{\ell_i(\x_i)}{\PP_i}
      		\sep \sum_{i\in I} \procout{\pp_i}{\ell_i}{\val_i}{\PP_i} & \text{\em (external choice, internal choice)}\\[1mm]
            & & \cond{\val} \PP \PQ & \text{\em (conditional)} \\[1mm]
            & & X 
            \sep \mu X.\PP 
            \sep \inact  & \text{\em (variable, recursion, inaction)}
  \end{array}
  \)
\caption{\label{tab:sessions}Syntax of sessions, processes, and queues.} %
\vspace{-4mm}%
\end{figure}
The syntax of 
the calculus is defined in Table~\ref{tab:sessions}. 
Values can be either variables ($\x, \y, \z$) or constants (positive integers or booleans).

Our \textbf{asynchronous multiparty sessions}
(ranged over by $\N, \N',\ldots$) 
represent a parallel composition of \textbf{participants} (ranged over by $\pp,\pq,\ldots$) assigned with their 
\textbf{process} $\PP$ and \textbf{output message queue} $\h$ %
(notation: $\pa\pp\PP \pc \pa\pp \h$). %
Here, $\pa\pp \h$ denotes that $\h$ is the output message queue of participant $\pp$, and the \textbf{queued message} $\left(\q,\ell(\val)\right)$ represents that participant $\pp$ has sent message labeled $\ell$ with payload $\val$ to $\q$. 
In the syntax of processes, 
the \textbf{external choice} $\sum_{i\in I}\procin{\pp_i}{\ell_i(\x_i)}{\PP_i}$ denotes receiving from participant $\pp_i$ a message labelled $\ell_i$ with a value that should replace variable $\x_i$ in $\PP_i$, for any $i\in I$. %
The \textbf{internal choice}  $\sum_{i\in I} \procout{\pp_i}{\ell_i}{\val_i}{\PP_i}$ denotes sending to participant $\pp_i$ a message labelled $\ell_i$ with value $\val_i$, and then continuing as $\PP_i$, for any $i\in I$; %
we will see that, when the internal choice has more than one branch, then one of them is picked nondeterministically. %
In both external and internal choices, we assume $(\pp_i, \ell_i) \not= (\pp_j, \ell_j)$, for all $i, j \in I$, such that $i\not= j$. The \textbf{conditional} construct, process \textbf{variable}, \textbf{recursion}, and \textbf{inaction} $\inact$ are standard; we will sometimes omit writing $\inact$.
As usual, we assume that recursion is guarded, i.e., in the process $\mu X.\PP$,
the variable $X$ can occur in $\PP$ only under an internal or external
choice.

We show how this syntax can be used to compactly model two federated learning protocols \cite{popovic2023,ProkicGKPPK23}: a centralized (\Cref{ex:centralised}) and a decentralized (\Cref{ex:decentralised}). First, we set up some notation. 

We define a \textbf{concurrent input} macro $\procinconcur{\PR_j}{j\in I}{\PQ}$ to represent a process that awaits a series of incoming messages arriving in arbitrary order:%
\[
\begin{array}{r@{\;}c@{\;}l@{\qquad}r@{\;}c@{\;}l}
  \PR   & \Coloneqq & \procinNoSuf{\pp}{\msgLabel{l}(x)} \sep \procinNoSuf{\pp}{\msgLabel{l}(x)}.\PRio
  &%
  \PRio & \Coloneqq & \procinNoSuf{\pp}{\msgLabel{l}(x)} \sep \procoutNoSuf{\pp}{\msgLabel{l}}{\val} \sep \procinNoSuf{\pp}{\msgLabel{l}(x)}.\PRio \sep \procoutNoSuf{\pp}{\msgLabel{l}}{\val}.\PRio
\end{array}
\]
\vspace{-7mm}
\[
\begin{array}{l}
\procinconcur{\PR_i}{i\in I}{\PQ} \;=\;
\sum_{i\in I} 
\PR_i.
\procinconcur{\PR_j}{j\in I{\setminus}\{i\}}{\PQ}
\quad \text{and} \quad 
\procinconcur{\PR_j}{j\in \{i\}}{\PQ} \;=\; \PR_i.\PQ
\end{array}
\]
where $\PR_i$ are all process prefixes starting with an input, possibly followed by other inputs or outputs.
The concurrent input $\procinconcur{\PR_i}{i\in I}{\PQ}$ 
specifies an external choice of any input-prefixed $\PR_i$, for $i\in I$, 
followed by another choice of input-prefixed process with an index from $I{\setminus}\{i\}$; the composition continues until all $i \in I$ are covered, with process $\PQ$ added at the end. In this way, concurrent input process $\procinconcur{\PR_i}{i\in I}{\PQ}$ can perform the actions specified by any $\PR_i$ (for $i\in I$) in an arbitrary order, depending on which inputs become available first; afterwards, process $\PQ$ is always executed. 
For instance, using our concurrent input macro, we may represent the FLA process 
$\PP_3$ from \Cref{sec:intro} as  
$
\PP_3 =
\procinconcurNoSuf{ 
\procinNoSuf{\pq}{\msgLabel{upd}(\x)},\,
\procinNoSuf{\pr}{\msgLabel{upd}(\y)}
}
$.

\begin{example}[Modelling a centralised FL protocol implementation]
  \label{ex:centralised} 
As specified in \cite{popovic2023,ProkicGKPPK23}, the \emph{generic centralized one-shot federated learning protocol}
is implemented with one server and $n-1$ clients (\Cref{sec:intro} presented the processes of the protocol for $n=3$). 
We may model an implementation of this protocol as a session $\M=\prod_{i\in I} (\pa{\pp_i}  \PP_i \;\pc\;
\pa{\pp_i}\EmptyQueue)$, where $I=\{1,2\ldots, n\}$, and where $\pp_1$ plays the role of the server, while the rest of participants are clients, defined with:
\begin{align*}
\PP_1 &= 
\procoutNoSuf{\pp_2}{\msgLabel{ld}}{\msgPayload{data}}.\ldots\procoutNoSuf{\pp_n}{\msgLabel{ld}}{\msgPayload{data}}.\procinconcurNoSuf{\procinNoSuf{\pp_2}{\msgLabel{upd}(x_2)}, \ldots, \procinNoSuf{\pp_n}{\msgLabel{upd}(x_n)}}
\\ 
\PP_i &= 
\procinNoSuf{\pp_1}{\msgLabel{ld}(x)}.\procoutNoSuf{\pp_1}{\msgLabel{upd}}{\msgPayload{data}} \qquad\text{ for } i=2,\ldots, n
\end{align*}
Notice that, after sending the data to all the clients, the server $\PP_1$ then receives the updates from all of them in an arbitrary order. The clients first receive the data and then reply the update back to the server.
\end{example}

\begin{example}[Modelling a decentralised FL protocol implementation]
  \label{ex:decentralised}
As specified in \cite{popovic2023,ProkicGKPPK23}, 
an implementation of the \emph{generic decentralized one-shot federated learning protocol}
comprises $n$ participant processes acting both as servers and clients (\Cref{sec:intro} presented one  process of the protocol for $n=3$).   
We may model an implementation of this protocol with a session $\M=\prod_{i\in I} (\pa{\pp_i}  \PP_i \;\pc\;
\pa{\pp_i}\EmptyQueue)$, where $I=\{1,2\ldots, n\}$, and where process $\PP_1$ is defined as follows: (the rest of the processes are defined analogously)
\begin{align*}
\PP_1 &= 
\procoutNoSuf{\pp_2}{\msgLabel{ld}}{\msgPayload{data}}.\ldots 
\procoutNoSuf{\pp_n}{\msgLabel{ld}}{\msgPayload{data}}.\\
& \hspace{5mm}
\procinconcurNoSuf{
\procin{\pp_2}{\msgLabel{ld}(\x_2)}{
\procoutNoSuf{\pp_2}{\msgLabel{upd}}{\msgPayload{data}}}, \ldots,  
\procinNoSuf{\pp_n}{\msgLabel{ld}(\x_n)}.
\procoutNoSuf{\pp_n}{\msgLabel{upd}}{\msgPayload{data}}
}. \\ 
& \hspace{5mm} 
\procinconcurNoSuf{
\procinNoSuf{\pp_2}{\msgLabel{upd}(\y_2)}, \ldots,  
\procinNoSuf{\pp_n}{\msgLabel{upd}(\y_n)}
}
\end{align*}
Process $\PP_1$ specifies that participant $\pp_1$ first sends its local data to all other participants. 
Then, $\pp_1$ receives local data from all other participants and replies the update concurrently (i.e., in an arbitrary order). 
Finally, $\pp_1$ receives the updates from all other participants, again in an arbitrary order. 
\end{example}

\begin{figure}[t]{}
\(%
\begin{array}{@{}ll@{}}
\rulename{r-send} &       
\pa\pp \sum_{i\in I}  \procout{\pq_i}{\ell_i}{\val_i}{\PP_i} \pc \pa\pp\h_{\pp} \pc \N \\ 
& \hspace{10mm} \redSend{\pp}{\pq_k}{\ell_k}\; \pa\pp\PP_k \pc \pa\pp\h_\pp\cdot(\q_k,\ell_k(\val_k)) \pc \N \qquad (k \in I)  
\\[1mm]
\rulename{r-rcv}
&
\pa\pp\sum_{i\in I} \procin{\pq_i}{\ell_i(\x_i)}\PP_i \pc \pa\pp\h_\pp\pc \pa{\q}\Q \pc \pa{\pq} (\pp,\ell(\val)) \cdot \h \pc \N
\\
& \hspace{5mm}\redRecv{\pp}{\pq}{\ell}\; \pa\pp \PP_k\subst{\val}{\x_k} \pc \pa\pp\h_\pp\pc \pa{\q}\Q \pc \pa{\pq} \h \pc \N \quad (\exists k{\in} I: (\pq, \ell)=(\pq_k, \ell_k))  
\\[1mm]
\rulename{r-cond-T} & 
\pa\pp{\cond{\true}{\PP}{\Q}}  \pc \pa\pp\h \pc \N \;\redIf{\pp}\; \pa\pp\PP \pc\pa\pp\h \pc  \N 
\\[1mm]
\rulename{r-cond-F} & 
\pa\pp{\cond{\false}{\PP}{\Q}}  \pc \pa\pp\h \pc \N \;\redIf{\pp}\; \pa\pp\Q \pc \pa\pp\h \pc \N
\\[1mm]
\rulename{r-struct} & 
\N_1\equivv \N_1' \;\;\;\text{and}\;\;\; \N_1'\red \N_2' \;\;\;\text{and}\;\;\; \N_2'\equivv\N_2 \quad\implies\quad 
\N_1\red\N_2
\end{array}
\)
\caption{\label{tab:main:reduction}Reduction relation on sessions.}%
\vspace{-4mm}
\end{figure}

\emph{Session Reductions.} The \textbf{reduction relation} for our process calculus
is defined in \Cref{tab:main:reduction}. 
Rule \rulename{r-send} specifies sending one of the messages from the internal choice, while the other choices are discarded;
the message $\ell_k$ with payload $\val_k$ sent to participant $\pq_k$ is appended to the participant $\pp$'s output queue, and $\pp$'s process becomes the continuation $\PP_k$. 
Dually, rule \rulename{r-rcv} defines how a participant $\pp$ can receive a message, directed towards $\pp$ with a supported label, 
from the head of the output queue of the sender $\pq$;
after the reduction, the message is removed from the queue and the received message then substitutes the placeholder variable $x_k$ in the continuation process $\PP_k$.
Observe that, by rules \rulename{r-send} and \rulename{r-rcv}, output queues are used in FIFO order. %
Rules \rulename{r-cond-t} and \rulename{r-cond-f} define how to reduce a conditional process: the former when the condition is true, the latter when it is false. 
Rule \rulename{r-struct} closes the reduction relation under a standard precongruence relation $\equivv$ (defined in \Cref{app:processes}) that allows reordering of messages in queues and unfolding recursive processes, combining  the approach of \cite{UdomsrirungruangY25} with \cite{GhilezanPPSY23}.

\emph{Properties.} In \Cref{def:session-properties} below
we follow the approach of  \cite{GhilezanPPSY23} and define how ``good'' sessions are expected to run by using behavioural properties.
We define three behavioural properties. 
The first one is \emph{safety}, ensuring that a session never has mismatches between the message labels supported by external choices
and the labels of incoming messages. Since, in our sessions, one participant can choose to receive messages from multiple senders at once,
our definition of safety requires external choices to support all possible message labels from all senders' queues.
The \emph{deadlock freedom} property requires that a session can get stuck (cannot reduce further) only in case it terminates.
The \emph{liveness} property ensures all pending inputs eventually receive a message and all queued messages are eventually received, 
if reduction steps are performed in a \emph{fair} manner. 

The fair path definition says that whenever a participant $\pp$ is ready to perform an output in a session $\M_i$, then there is a session $\M_k$ (reached by reducing $\M_i$) where $\pp$ actually performs an output. %
Dually, if a participant $\pp$ in $\M_i$ is ready to receive a message which is already available at the top of the sender’s output queue, then $\pp$ will receive that message in $\M_k$. 
This avoids unfair paths, e.g., in which two participants recursively exchange messages, while a third participant forever waits to send a message. 
Our fair and live properties for session types with standard choices matches the liveness defined in 
\cite[Definition 2.2]{GhilezanPPSY23}. Also, our fairness aligns with ``fairness of components'' according to \cite{GlabbeekH19} - where a ``component'' is a process in a session. Moreover, by the syntax and semantics of our processes, fairness of components coincides with justness \cite{GlabbeekH19}.

\begin{definition}[Session behavioral properties]
  \label{def:session-properties}%
  A \textbf{session path} %
  is a (possibly infinite) sequence of sessions %
  $(\M_i)_{i \in I}$, %
  where $I = \{0,1,2,\ldots, n\}$ (or $I = \{0,1,2,\ldots\}$) is a set of consecutive natural numbers, %
  and, $\forall i \in I{\setminus}\{n\}$ (or $\forall i \in I)$, $\M_i \red \M_{i+1}$. %
We say that a path $(\M_i)_{i \in I}$ is \textbf{safe} iff, %
  $\forall i \in I$:
  \begin{enumerate}[leftmargin=9mm,label={\sf\textbf{(SS)}},ref={\sf\textbf{SS}}]
  \item\label{item:session-safe}%
    if\; $\M_i \equivv \pa\pp\sum_{j\in J} \procin{\pq_j}{\ell_j(\x_j)}\PP_j  \pc \pa\pp\h \pc \pa\pq{\PQ} \pc \pa\pq{\h_\pq} \pc \M'$ with $\pq\in\{\pq_j\}_{j\in J}$ and $\h_\pq\equiv \msg{\pp}{\ell}{\val}\cdot\h_\pq'$, %
    \;then\; $(\pq, \ell)=(\pq_j, \ell_j)$, for some $j\in J$
  \end{enumerate}
\noindent
We say that a path $(\M_i)_{i \in I}$ is \textbf{deadlock-free} iff: 
  \begin{enumerate}[leftmargin=9mm,label={\sf\textbf{(SD)}},ref={\sf\textbf{SD}}]
  
  \item\label{item:session-deadlock-free}%
    if $I=\{0,1,2,\ldots, n\}$  and $\M_n\nred$ %
    \;then\; $\M_n \equivv  \pa\pp{\inact} \pc \pa\pp{\emptyqueue}$
  \end{enumerate}
  \noindent%
  We say that a path $(\M_i)_{i \in I}$ is \textbf{fair} iff, %
  $\forall i \in I$:
  \begin{enumerate}[leftmargin=9mm,label={\sf\textbf{(SF{\arabic*})}},ref={\sf\textbf{SF{\arabic*}}}]
  \item\label{item:session-fairness:send}%
    if $\M_i \redSend{\pp}{\pq}{\ell} \M'$, %
        then  $\exists k, j$ such that $I\ni k\geq  i$ %
    and $\M_k \redSend{\pp}{\pq_j}{\ell_j} \M_{k+1}$
  \item\label{item:session-fairness:recv}%
    if $\M_i \redRecv{\pp}{\pq}{\ell} \M'$, %
 $\exists k, j$ such that $I\ni k\geq  i$ %
    and $\M_k \redRecv{\pp}{\pq_j}{\ell_j} \M_{k+1}$
  \item\label{item:session-fairness:if}%
    if $\M_i \redIf{\pp} \M'$, %
      $\exists k$ such that $I\ni k\geq  i$ %
    and $\M_k \redIf{\pp} \M_{k+1}$
  \end{enumerate}
  \noindent%
  We say that a session path $(\M_i)_{i \in I}$ is \textbf{live} iff, %
  $\forall i \in I$:
  \begin{enumerate}[leftmargin=9mm,label={\sf\textbf{(SL{\arabic*})}},ref={\sf\textbf{SL{\arabic*}}}]
  \item\label{item:session-liveness:if}%
    if\; $\M_i \equivv \pa\pp{\cond{\val}{\PP}{\PQ}} \pc \pa\pp\h \pc \M'$, %
    \;then\; 
      $\exists k$ such that $I\ni k\geq  i$ %
  and,  for some $\M''$, we have either\; %
    $\M_k \equivv \pa\pp{\PP} \pc \pa\pp\h \pc \M''$
    \;or\; %
    $\M_k \equivv \pa\pp{\PQ} \pc \pa\pp\h \pc \M''$%
  \item\label{item:session-liveness:out}%
    if\; $\M_i \equivv \pa\pp\sum_{j\in J}{\procout{\pq_j}{\ell_j}{\val_j}{\PP_j}} \pc \pa\pp\h \pc \M'$, %
    \;then\;
      $\exists k$ such that $I\ni k\geq  i$ %
    \;and\; %
    $\M_k   \redSend{\pp}{\pq_j}{\ell_j} \M_{k+1}$, for some $j\in J$
  \item\label{item:session-liveness:send}%
    if\; $\M_i \equivv \pa\pp\PP \pc \pa\pp\msg{\pq}{\ell}{\val}\cdot\h \pc \M'$, %
    \;then\; 
      $\exists k$ such that $I\ni k\geq  i$ %
    \;and\; %
    $\M_k \redRecv{\pq}{\pp}{\ell} \M_{k+1}$
  \item\label{item:session-liveness:recv}%
    if\; $\M_i \equivv \pa\pp\sum_{j\in J} \procin{\pq_j}{\ell_j(\x_j)}\PP_j  \pc \pa\pp\h \pc \M'$, %
    \;then\; 
     $\exists k$ such that $I\ni k\geq  i$,  
    \;and\; %
    $\M_k \redRecv{\pp}{\pq}{\ell} \M_{k+1}$, where $(\pq, \ell)\in \{(\pq_j, \ell_j)\}_{j\in J}$
  \end{enumerate}
  
  \noindent%
  We say that a session $\M$ is \textbf{safe}/\textbf{deadlock-free} %
  iff  %
  all paths beginning with $\M$ %
  are safe/deadlock-free. %
  We say that a session $\M$ is \textbf{live} %
  iff %
  all fair paths beginning with $\M$ %
  are live. 
\end{definition}

Like in standard session types \cite{ScalasY19}, in our calculus safety does not imply liveness nor deadlock-freedom (\Cref{ex:safe-but-not-live-session}).
Also, liveness implies deadlock-freedom, since liveness requires that session cannot get stuck before all external choices are performed and queued messages are eventually received. The converse is not true: e.g., a session in which two participants recursively exchange messages is deadlock free (as it never gets stuck), even if a third participant waits forever to receive a message that nobody will send. 
This session is not live, since in any fair path the third participant never fires its actions. 

However, \emph{unlike} standard session types, in our calculus liveness does \emph{not} imply safety: %
this is illustrated in \Cref{ex:live-but-not-safe-session} below.%

\begin{example}[A safe but non-live session]
\label{ex:safe-but-not-live-session}
Consider session 
\begin{align*}
\M &=  \pa{\pp}  \sum \{
\procin{\pq}{\ell_1(\x)}{\procinNoSuf{\pr}{\ell_2(\y)}}, 
\; \procinNoSuf{\pr}{\ell_2(\x)}, 
\; \procinNoSuf{\pr}{\ell_3(\x)}
\} \;\pc\;
\pa{\pp}\EmptyQueue \\
&  \quad\;\pc\;
\pa{\pq}\inact
\;\pc\;
\pa{\pq}\msg{\pp}{\ell_1}{\val_1}
\;\pc\;
\pa{\pr}\inact
\;\pc\;
\pa{\pr}\msg{\pp}{\ell_2}{\val_2}
\end{align*}
In the session, participant $\pp$ is ready to receive an input either from $\pq$ or $\pr$, 
which, in turn, both have enqueued messages for $\pp$. 
The labels of enqueued messages are safely 
supported in $\pp$'s receive. 
Session $\M$ has two possible reductions, where $\pp$ receives   
either from $\pq$, and    
$
\M \red \pa{\pp}  
\procinNoSuf{\pr}{\ell_2(\y)}
\;\pc\;
\pa{\pp}\EmptyQueue 
\;\pc\;
\pa{\pr}\inact
\;\pc\;
\pa{\pr}\msg{\pp}{\ell_2}{\val_2}
\red 
\pa{\pp}  
\inact
\;\pc\;
\pa{\pp}\EmptyQueue 
$;  
or from $\pr$, in which case  
$
\M \red  
\pa{\pq}\inact
\;\pc\;
\pa{\pq}\msg{\pp}{\ell_1}{\val_1}
$.  

Hence, session $\M$ is safe, since in all reductions inputs and matching enqueued messages have matching labels. 
However, $\M$ is not live, since the second path above starting with $\M$ is not live, for $\pq$'s output queue contains an orphan message that cannot be received.
Notice also that $\M$ is not deadlock-free since the final session in the second path cannot reduce further but is not equivalent to a terminated session.
\end{example}

\begin{example}[A live but unsafe session]
\label{ex:live-but-not-safe-session} 
Consider the following session:
\begin{align*}
\M' &=  \pa{\pp}  \sum \{
\procin{\pq}{\ell_1(\x)}{\procinNoSuf{\pr}{\ell_2(\y)}}, \; \procinNoSuf{\pr}{\ell_3(\x)}
\} \;\pc\;
\pa{\pp}\EmptyQueue \\
&  \quad\;\pc\;
\pa{\pq}\inact
\;\pc\;
\pa{\pq}\msg{\pp}{\ell_1}{\val_1}
\;\pc\;
\pa{\pr}\inact
\;\pc\;
\pa{\pr}\msg{\pp}{\ell_2}{\val_2}
\end{align*}
The session $\M'$ is not safe since the message in $\pr$'s output queue has label $\ell_2$ which is not supported in $\pp$'s
external choice (that only supports label $\ell_3$ for receiving from $\pr$). However, $\M'$ has a single reduction path
where $\pp$ receives from $\pq$, and then $\pp$ receives 
from $\pr$; 
then, the session ends with
all processes being $\inact$ and all queues empty, which implies $\M'$ is live (and thus, also deadlock-free).
\end{example}

\section{Types and Typing Environments}\label{sec:tsu}
We now introduce our (local) session types, 
typing contexts and their properties, and subtyping. Our types are a blend of asynchronous multiparty session types \cite{GhilezanPPSY23}
and the separated choice multiparty sessions (SCMP) types from
\cite{PetersY24}. 

\begin{definition}\label{def:types}\label{def:sorts}\label{def:typing-env} 
The \textbf{sorts $\ST$} are defined as 
$\ST \Coloneqq \tnat  \!\!\sep\!\! \tbool$, 
and \textbf{(local) session types $\T$} are defined as:
\[
  \begin{array}{@{}r@{\;\;}c@{\;\;}l@{\qquad}r@{\;\;}c@{\;\;}l@{}}
    \T  &\Coloneqq &   \tinternal_{i\in I}\tout{\pp_i}{\ell_i}{\ST_i}.\T_i 
    \sep \texternal_{i\in I}\tin{\pp_i}{\ell_i}{\ST_i}.\T_i   \sep \tend  \sep      \mu\ty. \T   \sep     \ty 
  \end{array}
\]
where $(\pp_i, \ell_i) \not= (\pp_j, \ell_j)$, for all $i, j \in I$ such that $i \neq j$. The \textbf{queue types $\tqueue$} and \textbf{typing environments $\Gamma$} are defined as: 
\[
  \begin{array}{@{}r@{\;\;}c@{\;\;}l@{\qquad\qquad}r@{\;\;}c@{\;\;}l@{}}
     \tqueue &\Coloneqq&
 \temptyqueue \sep \tout\pp\ell\ST \sep \tqueue \cdot \tqueue &
    \Gamma &\Coloneqq& \emptyset \sep \Gamma, \pp:(\tqueue, \T)
  \end{array}
\]
\end{definition}

\emph{Sorts} are the types of values, which can be natural numbers ($\tnat$) or booleans ($\tbool$). 
The \emph{internal choice} session type $\tinternal_{i\in I}\tout{\pp_i}{\ell_i}{\ST_i}.\T_i$ describes an output of message $\ell_i$ with sort $\ST_i$ towards participant $\pp_i$ and then evolving to type $\T_i$, for some $i\in I$. 
Similarly, $\texternal_{i\in I}\tin{\pp_i}{\ell_i}{\ST_i}.\T_i$ stands for \emph{external choice}, i.e., receiving a message $\ell_i$ with sort $\ST_i$ from participant $\pp_i$, for some $i\in I$. 
The type $\tend$ denotes the terminated session type,
$\mu\ty. \T$ is a recursive type, and 
$\ty$ is a recursive type variable. 
We assume that all recursions are guarded and follow a form of Barendregt convention: every $\mu\ty.T$ binds a syntactically distinct $\ty$. 

The \emph{queue type} represents the type of the messages contained in an output queue; it can be empty ($\temptyqueue$), or it can contain message $\ell$ with sort $\ST$ for participant $\pp$ ($\tout\pp\ell\ST$), or it can be a concatenation of two queue types ($\tqueue \cdot \tqueue$). 
The \emph{typing environment} assigns a pair of queue/session type to a participant ($\pp:(\tqueue, \T)$). 
We use $\Gamma(\pp)$ to denote the type that $\Gamma$ assigns to $\pp.$

\emph{Typing Environment Reductions.} Now we define 
typing environment reductions,  
which relies on a structural congruence relation $\equiv$ %
over session types, queue types, and typing environments (defined in \Cref{app:types}).

\begin{definition}
\label{def:typing-env-reductions}%
  The typing environment reduction\,  $\redLabel{\;\alpha\;}$, with $\alpha$ being either\,
    $\recvLabel{\pp}{\pq}{\ell}$
    \,or $\sendLabel{\pp}{\pq}{\ell}$
    \,(for some $\pp,\pq,\ell$), %
  is inductively defined as follows:

\smallskip%
\noindent%
\centerline{%
\small\rm%
\(%
  \begin{array}{@{}l@{\;}l}
\rulename{e-rcv} &
{\pp_k}{:}\,(\tout{\pq}{\ell_k}{\ST_k}{\cdot} \tqueue, \T_{\pp}),\pq{:}\,(\tqueue_{\pq}, \texternal_{i\in I}\tin{\pp_i}{\ell_i}{\ST_i}.{\T_i}), \Gamma \\
&\hspace{10mm} \redRecv{\pq}{\pp_k}{\ell_{k}}\; {\pp_k}{:}\,(\tqueue, \T_\pp), \pq{:}\,(\tqueue_{\pq}, \T_k),\Gamma %
\qquad (k\!\in\! I)\\
\rulename{e-send}
   & {\pp:(\tqueue, \tinternal_{i\in I}\tout{\pq_i}{\ell_i}{\ST_i}.{\T_i}),\Gamma \;\;\redSend{\pp}{\pq_k}{\ell_{k}}\;\; \pp:(\tqueue{\cdot} \tout{\pq_k}{\ell_k}{\ST_k}\phantom{}, 
   \T_k),\Gamma} \qquad
   {(k\in I) }\\[1mm]
\rulename{e-struct}&
\quad \Gamma \equiv \Gamma_1 \redLabel{\;\alpha\;} \Gamma_1' \equiv \Gamma'%
\;\implies\;%
\Gamma \redLabel{\;\alpha\;} \Gamma'%
   \end{array}
\)%
}

\noindent%
We use $\Gamma \red \Gamma'$ %
  instead of $\Gamma \redLabel{\;\alpha\;} \Gamma'$ %
  when $\alpha$ is not relevant, and 
  $\reds$ for the reflexive and transitive closure of $\red$. $\Gamma \red$ denotes $\Gamma\red \Gamma'$, for some $\Gamma'$.
\end{definition}

In rule \rulename{e-rcv} an environment has a reduction step labeled 
${\pq}{:}{\pp_k}{?}{\ell_{k}}$ if %
participant $\pp_k$ has at the head of its queue a message for $\pq$ %
with label $\ell_k$ and payload sort $\ST_k$, %
and $\pq$ has an external choice type that includes participant $\pp_k$ 
with label $\ell_k$ and a corresponding sort $\ST_k$; %
the environment evolves by consuming $\pp_k$'s message %
and activating the continuation $\T_k$ in $\pq$'s type. %
Rule \rulename{e-send} %
specifies reduction if $\pp$ %
has an internal choice type where $\pp$ sends a message %
toward $\pq_k$, for some $k\in I$; the reduction is labelled ${\pp}{:}{\pq_k}{!}{\ell_{k}}$ %
and the message is placed at the end of $\pp$'s queue. %
Rule \rulename{e-struct} is standard closure of the reduction relation %
under structural congruence.

\emph{Properties.} Similarly to~\cite{ScalasY19}, %
we define behavioral properties also for typing environments %
(and their reductions). 
As for processes, we use three properties for typing environments. 
\emph{Safety} ensures that the typing environment never has label or sort mismatches. 
In our setting, where one participant can receive from more than one queue, 
this property ensures safe receptions of queued messages.
The second property, \emph{deadlock freedom},
ensures that typing environment which can not reduce further must be terminated.
The third property is \emph{liveness}, which,
for the case of standard session types (as in \cite{GhilezanPPSY23}), matches the liveness defined in 
\cite[Definition 4.7]{GhilezanPPSY23}:
it ensures all pending inputs eventually receive a message and all queued messages are eventually received, 
if reduction steps are performed \emph{fairly}.

Notably, our definition of deadlock-freedom and liveness also require safety.
The reason for this is that liveness and deadlock-freedom properties for typing environments, if defined without assuming safety, are not preserved by the
subtyping (introduced in the next section, see
\Cref{ex:dfree-and-live-not-preserved-subtyping}). We use the predicate over
typing environments $\fend(\Gamma)$ (read ``\textbf{$\Gamma$ is terminated}'')
that holds iff, for all $p\in\dom{\Gamma}$, we have $\Gamma(\pp) \!\equiv\!
(\temptyqueue,\tend)$.

\begin{definition}[Typing environment properties]
  \label{def:env-path-properties}%
  A \emph{typing environment path} %
  is a 
  (possibly infinite) sequence of typing environments %
  $(\Gamma_i)_{i \in I}$, %
  where $I = \{0,1,2,\linebreak\ldots, n\}$ (or $I = \{0,1,2,\ldots\}$) is a set of consecutive natural numbers, %
  and, $\forall i \in I{\setminus}\{n\}$ (or $\forall i \in I$), $\Gamma_i \red \Gamma_{i+1}$. %
    We say a path $(\Gamma_i)_{i \in I}$ is \textbf{safe} iff, %
  $\forall i \in I$:
  \begin{enumerate}[leftmargin=9mm,label={\sf\textbf{(TS)}},ref={\sf\textbf{TS}}]
  \item\label{item:type:path:safe}%
   if $\Gamma_i(\pp)\equiv (\tqueue_\pp, \tinternal_{j\in J}\tin{\pq_j}{\ell_j}{\ST_j}.{\T_j})$ and $\Gamma_i(\pq_k)\equiv (\tout{\pp}{\ell_k}{\ST_k}{\cdot} \tqueue_\pq, \T_{\pq})$, %
   with $\pq_k\in\{\pq_j\}_{j\in J}$, %
   then 
   $\exists j\in J: (\pq_j, \ell_j, \ST_j)=(\pq_k, \ell_k, \ST_k)$
  \end{enumerate}
    We say that a path $(\Gamma_i)_{i \in I}$ is \textbf{deadlock-free} iff: 
  \begin{enumerate}[leftmargin=9mm,label={\sf\textbf{(TD)}},ref={\sf\textbf{TD}}]
  \item\label{item:type:path:deadlock-free}%
    if $I=\{0,1,\ldots,n\}$ and $\Gamma_n\nred$ %
    then $\fend(\Gamma_n)$
    \end{enumerate}
We say that a path $(\Gamma_i)_{i \in I}$ is \textbf{fair} iff, %
  $\forall i \in I$:
  \begin{enumerate}[leftmargin=9mm,label={\sf\textbf{(TF{\arabic*})}},ref={\sf\textbf{TF{\arabic*}}}]
  \item\label{item:fairness:send}%
    if $\Gamma_i \redSend{\pp}{\pq}{\ell} \Gamma'$, %
    then 
      $\exists k, \pq', \ell'$ such that $I\ni k\geq  i$ %
    and $\Gamma_k \redSend{\pp}{\pq'}{\ell'} \Gamma_{k+1}$
  \item\label{item:fairness:recv}%
    if $\Gamma_i \redRecv{\pp}{\pq}{\ell} \Gamma'$, %
    then 
      $\exists k, \pq', \ell'$ such that $I\ni k\geq  i$ %
    and $\Gamma_k \redRecv{\pp}{\pq'}{\ell'} \Gamma_{k+1}$
  \end{enumerate}

  \noindent%
  We say that a path $(\Gamma_i)_{i \in I}$ is \textbf{live} iff, %
  $\forall i \in I$:
  \begin{enumerate}[leftmargin=9mm,label={\sf\textbf{(TL{\arabic*})}},ref={\sf\textbf{TL{\arabic*}}}]
  \item\label{item:liveness:send}%
    if $\Gamma_i(\pp) \equiv \left(\tout\pq{\ell}\ST \cdot \tqueue \,,\, \T\right)$, %
    then
    $\exists k$ such that $I\ni k\geq  i$ %
    \;and\; %
    $\Gamma_k \redRecv{\pq}{\pp}{\ell} \Gamma_{k+1}$
  \item\label{item:liveness:recv}%
    if $\Gamma_i(\pp) \equiv \left(\tqueue_\pp \,,\, \texternal_{j \in J}{\tin{\pq_j}{\ell_j}{\ST_j}.{\T_j}}\right)$, %
    then 
    $\exists k, \pq, \ell$ such that $I\ni k\geq  i$, 
    $(\pq,\ell)\in \{(\pq_j,\ell_j)\}_{j\in J}$,%
    \;and\; %
    $\Gamma_k \redRecv{\pp}{\pq}{\ell} \Gamma_{k+1}$
  \end{enumerate}
  
  \noindent%
  A typing environment $\Gamma$ is \textbf{safe} iff all paths starting with $\Gamma$ are safe. 
  A typing environment $\Gamma$ is \textbf{deadlock-free} iff all paths starting with $\Gamma$ are safe and deadlock-free. 
  We say that a typing environment $\Gamma$ is \textbf{live} %
  iff %
  it is safe and all fair paths beginning with $\Gamma$ %
  are live. 
\end{definition}

Since our deadlock-freedom and liveness for typing environments subsumes safety, we do not have the situation in which typing environment is deadlock-free and/or live but not safe. 
Still, we can have typing environments that are safe but not deadlock-free or live. 

\begin{example}
\label{ex:typing-environment-properties}
Consider typing environment 
\begin{align*}
\Gamma &= \{ 
\pp{:}\,(\temptyqueue, 
\sum 
\{
\tin{\pq}{\ell_1}{\ST_1}.\tin{\pr}{\ell_2}{\ST_2}, 
\; \tin{\pr}{\ell_2}{\ST_2}, 
\; \tin{\pr}{\ell_3}{\ST_3}
\}, \\
& \hspace{10mm}{\pq}{:}\,(\tout{\pp}{\ell_1}{\ST_1}, \tend), 
{\pr}{:}\,(\tout{\pp}{\ell_2}{\ST_2}, \tend) 
\}
\end{align*}
Here, $\Gamma$ is safe by \Cref{def:env-path-properties},
but is not live nor deadlock-free, because the path in which $\pp$ receives from $\pr$ message $\ell_2$ leads to
a typing environment that cannot reduce further but is not terminated.
Now consider typing environment 
\begin{align*}
\Gamma' = \{ 
\pp{:}\,(\temptyqueue, 
\sum 
\{
\tin{\pq}{\ell_1}{\ST_1}.\tin{\pr}{\ell_2}{\ST_2}, 
\; \tin{\pr}{\ell_3}{\ST_3}
\}, 
{\pq}{:}\,(\tout{\pp}{\ell_1}{\ST_1}, \tend), 
{\pr}{:}\,(\tout{\pp}{\ell_2}{\ST_2}, \tend) 
\}
\end{align*}
$\Gamma'$ is not safe by \Cref{def:env-path-properties} (hence, also not deadlock-free nor live)
since there is a mismatch between $\pr$ sending $\ell_2$ to $\pp$, while $\pp$
can only receive $\ell_3$ from $\pr$. 
\end{example}

We now prove that all three behavioral properties are closed under structural congruence \emph{(in \Cref{app:subtyping})} and reductions \emph{(proof in \Cref{app:subtyping})}.

\begin{restatable}{proposition}{PropRedPresevesProperties}
\label{lem:subtyping-and-reduction}
  \label{lem:move-preserves-safety}%
  If\, $\Gamma$ is safe/deadlock-free/live and $\Gamma \red \Gamma'$, %
  then $\Gamma'$ is safe/ deadlock-free/live.
\end{restatable}

\subsection{Subtyping}
\label{sec:subtyping}

In \Cref{def:subtyping} below we formalise a standard ``synchronous'' subtyping
relation as in \cite{PetersY24}, not dealing with the possible reorderings of outputs and inputs allowed by asynchronous subtyping \cite{GhilezanPPSY23}.

\begin{definition}
\label{def:subtyping}%
  \;The \emph{subtyping} $\subt$ %
  is coinductively defined as:
\[
  \small
\begin{array}{@{}c@{}}
 \cinfer[\rulename{s-out}]{\forall i \in I \quad \T_i \subt \T_i' \quad \{\pp_i\}_{i\in I}=\{\pp_i\}_{i\in I \cup J}}{
  \tinternal_{i\in I}\tout{\pp_i}{\ell_i}{\ST_i}.\T_i
  \subt
   \tinternal_{i\in I\cup J}\tout{\pp_i}{\ell_i}{\ST_i}.\T_i'
 }
 \\[2mm]
\cinfer[\rulename{s-in}]{\forall i \in I \quad \T_i \subt \T_i' \quad \{\pp_i\}_{i\in I}=\{\pp_i\}_{i\in I \cup J}}{
  \texternal_{i\in I\cup J}\tin{\pp_i}{\ell_i}{\ST_i}.\T_i
  \subt
   \texternal_{i\in I}\tin{\pp_i}{\ell_i}{\ST_i}.\T_i'
 }
\\[2mm]
 \cinferrule[\rulename{s-muL}]{
 \T_1\subst{\mu \ty.\T_1}{\ty} \subt \T_2}
 {\mu\ty.\T_1 \subt \T_2}
 \qquad
  \cinferrule[\rulename{s-muR}]{
 \T_1 \subt \T_2\subst{\mu \ty.\T_2}{\ty}}
 {\T_1 \subt \mu\ty.\T_2}
 \qquad
  \cinferrule[\rulename{s-end}]{}
 {\tend \subt \tend}
\end{array}
\]
Pair of queue/session types are related via subtyping\; $(\tqueue_1,\T_1)\subt (\tqueue_2, \T_2)$ \;iff\; 
$\tqueue_1=\tqueue_2$ and $\T_1 \subt \T_2$. 
We define $\Gamma\subt\Gamma'$ \;iff\; $\dom\Gamma=\dom{\Gamma'}$ and $\forall \pp \in\dom\Gamma: \Gamma(\pp)\subt\Gamma'(\pp)$.
\end{definition}

The subtyping rules allow less branches for the subtype in the internal choice (in \rulename{s-out}), and more branches in the external choice (in \rulename{s-in}). 
The side condition in both rules ensures the set of participants specified in the subtype and supertype choices is the same. 
Hence, subtyping allows flexibility in the set of labels, not in the set of participants. 
Subtyping holds up to unfolding (by \rulename{s-muL} and \rulename{s-muR}), as usual for coinductive subtyping \cite[Chapter 21]{PierceBC:typsysfpl}. 
By rule \rulename{s-end}, $\tend$ is related via subtyping to itself.

\begin{example}
\label{ex:subtype}
Consider the typing environments $\Gamma$ and $\Gamma'$ from \Cref{ex:typing-environment-properties}. Since
\(
\sum \{
\tin{\pq}{\ell_1}{\ST_1}.\tin{\pr}{\ell_2}{\ST_2}, 
\; \tin{\pr}{\ell_2}{\ST_2}, 
\; \tin{\pr}{\ell_3}{\ST_3}
\}
\;\subt\; 
\sum 
\{
\tin{\pq}{\ell_1}{\ST_1}.\tin{\pr}{\ell_2}{\ST_2}, 
\; \tin{\pr}{\ell_3}{\ST_3}
\}
\) 
holds by \rulename{s-in}, we also have that $\Gamma \subt \Gamma'$ holds.
\end{example}

The following two lemmas show that subtyping of safe typing environments is a simulation and that subtyping  preserves all typing environment properties.
\emph{(Proofs in \Cref{app:subtyping}.)}

\begin{restatable}{lemma}{LemSubtypingAndReduction}
\label{lem:subtyping-and-reduction}
If $\Gamma'\subt \Gamma$, $\Gamma$ is safe, and $\Gamma' \redLabel{\;\alpha\;} \Gamma'_1$, then there is $\Gamma_1$ such that $\Gamma'_1\subt \Gamma_1$ and $\Gamma \redLabel{\;\alpha\;} \Gamma_1$.
\end{restatable}

\begin{restatable}{lemma}{lemSubtypingPreservesSafety}
  \label{lem:subtyping-preserves-safety}%
  If\, $\Gamma$ is safe/deadlock-free/live 
  and $\Gamma' \subt \Gamma$, %
  then $\Gamma'$ is safe/deadlock-free/live.
\end{restatable}

\begin{remark}
If in \Cref{def:env-path-properties} safety is not assumed in deadlock-freedom and liveness, then \Cref{lem:subtyping-preserves-safety} does not hold (see \Cref{ex:subtype}). 
The same observation also holds for deadlock-freedom in the synchronous case \cite[Remark 3.2]{PetersY24}.
\end{remark}

\section{The Typing System and Its Properties}\label{sec:tsy}
\label{sec:type-system}

This section introduces the type system 
that assigns types 
to processes, queues, and sessions. 
Our typing system is an extension and adaptation of the one in \cite{GhilezanPPSY23}. 

\begin{definition}[Typing system]
\label{def:type-system}%
A \textbf{\emph{shared typing environment} $\Theta$}, which assigns sorts to expression variables, and (recursive) session types to process variables, is defined as 
$%
   \; \Theta \;\;\Coloneqq\;\; \emptyset \sep \Theta, X:\T \sep \Theta, \x:\ST
$. 

Our type system is inductively defined by the rules 
in \Cref{figure:typesystem}, with 4 judgements --- which cover respectively
values and variables $\val$, processes $\PP$, message queues $\h$, and sessions $\N$:
\[
\Theta \vdash \val:\ST \qquad\qquad
\Theta \vdash \PP:\T \qquad\qquad
\vdash \h:\tqueue \qquad\qquad
\Gamma \vdash \N
\]
\end{definition}

By \Cref{def:type-system}, the typing judgment $\Theta \vdash \val:\ST$ means that value $\val$ is of sort $\ST$ under environment $\Theta$.
The judgement $\Theta \vdash \PP:\T$ says that process $\PP$ behaves as prescribed with type $\T$ and uses its variables as given in $\Theta$. 
The judgement $\vdash \h:\tqueue$ says that the messages in the queue $\h$ correspond to the queue type $\tqueue$. 
Finally, $\Gamma \vdash \N$ means that all participant processes in session $\N$ behave as prescribed by the session types in $\Gamma$.

\begin{figure}[t!]
\(
\begin{array}{@{}c@{}}
 \inferrule[\rulename{t-nat}]{ }{\Theta \vdash 1,2,\ldots: \tnat}
 \quad
 \inferrule[\rulename{t-Bool}]{ }{\Theta\vdash \true, \false:\tbool} \quad
 \inferrule[\rulename{t-var}]{ }{\Theta, x:\ST \vdash x:\ST}
  \\[2mm]
    \inferrule[\rulename{t-nul}]{ }{\vdash \emptyqueue: \temptyqueue}
       \qquad
   \inferrule[\rulename{t-elm}]{\ \vdash\val:\ST}{\vdash  (\pq,\ell(\val)):\tout\pq\ell\ST}
       \qquad
   \inferrule[\rulename{t-queue}]{\vdash \h_1:\tqueue_1 \quad \vdash \h_2:\tqueue_2}{\vdash \h_1 \cdot \h_2: \tqueue_1\cdot\tqueue_2}\\[2mm]
  \inferruleR[\rulename{t-$\inact$}]{$\;$}{\Theta \vdash \inact: \tend}
  \qquad
\inferruleR[\rulename{t-out}]{ \forall i\in I\;\;\; \Theta \vdash \val_i:\ST_i \quad \Theta \vdash \PP_i:\T_i }{\Theta  \vdash \sum_{i\in I}\procout{\pp_i} {\ell_i}{\val_i}{\PP_i}: \tinternal_{i\in I}\tout{\pp_i}{\ell_i}{\ST_i}.{\T_i}}
  \\[2mm]%
\inferruleR[\rulename{t-in}]{\forall i\in I\;\;\; \Theta, x_i:\ST_i \vdash \PP_i:\T_i}{\Theta \vdash  \sum_{i\in I}\procin{\pp_i}{\ell_i(\x_i)}{\PP_i}: \texternal_{i\in I}\tin{\pp_i}{\ell_i}{\ST_i}.{\T_i}}
  \\[2mm]
\inferruleR[\rulename{t-cond}]{\Theta \vdash \val :\tbool
\quad \Theta  \vdash \PP_i:\T \ \text{\tiny $(i=1,2)$} 
}{\Theta \vdash \cond\val{\PP_1}{\PP_2}:\T}
\\[2mm]
  \inferruleR[\rulename{t-rec}]{\Theta, X:\T  \vdash \PP:\T} 
  {\Theta  \vdash  \mu X.\PP:  \T}
  \qquad
  \inferruleR[\rulename{t-var}]{$\;$}{\Theta, X:\T  \vdash X:\T}
  \qquad
  \inferruleR[\rulename{t-sub}]{\Theta \vdash \PP:\T \quad \T\subt \T' }{\Theta \vdash \PP:\T'}
\\[2mm]
\inferruleR[\rulename{t-sess}]{\Gamma=\{\pp_i{:}\,(\tqueue_i, \T_i) \;|\; i\in I\} \qquad %
  \forall i\in I 
  \qquad \emptyset \vdash \PP_i:\T_i \qquad  
  \vdash \h_i:\tqueue_i
}{
  \Gamma \vdash \prod_{i\in I} (\pa{\pp_i}{\PP_i} \pc \pa{\pp_i}{\h_i})
}
\end{array}
\)

\caption{\label{figure:typesystem}%
  Typing rules for values, queues,
  and for processes and sessions.%
}%
\vspace{-4mm}
\end{figure}

We now illustrate the typing rules in \Cref{figure:typesystem}. 
The first row gives rules for typing natural and boolean values and expression variables. 
The second row provides rules for typing message queues: an empty queue has the empty queue type (by \rulename{t-nul})
while a queued message is typed if its payload has the correct sort (by \rulename{t-elm}), and a queue obtained with concatenating two message queues is typed by concatenating their queue types (by \rulename{t-queue}).

Rule \rulename{t-$\inact$} types a terminated process $\inact$ with $\tend$. 
Rules \rulename{t-out}/\rulename{t-in} type internal/external choice processes with the corresponding internal/external choice types: 
in each branch the payload $\val_i$/$x_i$ and the continuation process $\PP_i$ have the corresponding sort $\S_i$ and continuation type $\T_i$;
observe that in \rulename{t-in}, 
each continuation process $\PP_i$ is typed 
under environment $\Theta$ extended with the input-bound variable $x_i$ having sort $\S_i$.
Rule \rulename{t-cond} types conditional: both branches must have the same type, 
and the condition must be boolean.
Rule \rulename{t-rec} types a recursive process $\mu X.\PP$ with $\T$ if $\PP$ has type $\T$ in an environment where that $X$ has type $\T$.
Rule \rulename{t-var} types recursive process variables. 
Finally, \rulename{t-sub} is the \emph{subsumption rule}: a process $\PP$ of type $\T$ can be typed with any supertype $\T'$. 
This rule makes our type system
equi-recursive \cite{PierceBC:typsysfpl}, as 
 $\subt$ relates types up-to unfolding (by \Cref{def:subtyping}).
 Rule \rulename{t-sess} types a session under environment $\Gamma$ if each participant's queue and process have corresponding queue and process types
 in an empty $\Theta$ (which forces processes to be closed).

\begin{example}[Types for a centralised FL protocol implementation]
  \label{ex:centralised-types} 
Consider session $\M$ from \Cref{ex:centralised}. 
Take that value $\msgPayload{data}$ and variables $\x, \x_2, \ldots, \x_n$ are of sort $\ST$. 
We may derive $\vdash \PP_1 : \T_1$ and $\vdash \PP_i: \T_2$ (for $i=1,2,\ldots,n$), where (with a slight abuse of notation using the concurrent input macro  for types):
\begin{align*}
\T_1 &= 
\procoutNoSuf{\pp_2}{\msgLabel{ld}}{\ST}.\ldots\procoutNoSuf{\pp_n}{\msgLabel{ld}}{\ST}.\procinconcurNoSuf{\procinNoSuf{\pp_2}{\msgLabel{upd}(\ST)}, \ldots, \procinNoSuf{\pp_n}{\msgLabel{upd}(\ST)}}
\\ 
\T_2 &= 
\procinNoSuf{\pp_1}{\msgLabel{ld}(\ST)}.\procoutNoSuf{\pp_1}{\msgLabel{upd}}{\ST} 
\end{align*}
Hence, with $\Gamma =\{\pp_1:(\temptyqueue, \T_1)\}\cup \{\pp_i:(\temptyqueue, \T_2)\,|\, i=2,\ldots,n\}$ we obtain $\Gamma\vdash \M$.  
Observe that $\Gamma$ is deadlock-free and live: after $\pp_1$ sends the data to all other participants, each participant receives and replies with an update, that $\pp_1$ finally receives in arbitrary order.

Furthermore, considering again the upgrade process $\PQ'$ from \Cref{sec:intro} that allows multi-model FL. We may derive $\vdash \PQ': \T'_2$, for 
\[
\T'_2= \sum\{\procinNoSuf{\pp_1}{\msgLabel{ld}(\ST)}.\procoutNoSuf{\pp_1}{\msgLabel{upd}}{\ST}, \, \procinNoSuf{\pp_1}{\msgLabel{ld'}(\ST)}.\procoutNoSuf{\pp_1}{\msgLabel{upd'}}{\ST}\}
\] 
By our subtyping relation we have $\T_2' \subt \T_2$. 
If we denote with $\Gamma'$ typing environment where one of the clients has the upgraded type $\T_2'$, and the rest is same as in $\Gamma$, we also have $\Gamma' \subt \Gamma$. Hence, by \Cref{lem:subtyping-preserves-safety} we have that $\Gamma'$ is also deadlock-free and live, confirming that $\PQ$ can safely be replaced with $\PQ'$ in this, but also in any other, session. 
\end{example}

\begin{example}[Types for a decentralised FL protocol implementation]
  \label{ex:decentralised-types}
  Consider session $\M$ from \Cref{ex:decentralised}. 
Take that value $\msgPayload{data}$ and variables $\x_2, \ldots, \x_n, \y_2, \ldots, \y_n$ are of sort $\ST$.  
We may derive $\vdash \PP_1 : \T_1$ where (again with a slight abuse of notation using the concurrent input macro for types):
\begin{align*}
\T_1 &= 
\procoutNoSuf{\pp_2}{\msgLabel{ld}}{\ST}.\ldots 
\procoutNoSuf{\pp_n}{\msgLabel{ld}}{\ST}.
\procinconcurNoSuf{
\procin{\pp_2}{\msgLabel{ld}(\ST)}{
\procoutNoSuf{\pp_2}{\msgLabel{upd}}{\ST}}, \ldots,  
\procinNoSuf{\pp_n}{\msgLabel{ld}(\ST)}.
\procoutNoSuf{\pp_n}{\msgLabel{upd}}{\ST}
}. \\ 
& \hspace{10mm} 
\procinconcurNoSuf{
\procinNoSuf{\pp_2}{\msgLabel{upd}(\ST)}, \ldots,  
\procinNoSuf{\pp_n}{\msgLabel{upd}(\ST)}
}
\end{align*}
and similarly for other process we may derive the corresponding types so that $\vdash \PP_i:\T_i$, for $i=1, 2, \ldots, n$. 
Thus, with $\Gamma = \{\pp_i:(\temptyqueue, \T_i)\,|\, i=1,\ldots,n\}$ we obtain $\Gamma\vdash \M$. 
Notice that $\Gamma$ is safe, i.e., has no label or sort mismatches: participant $\pp_i$ always receives from $\pp_j$ message $\msgLabel{ld}$ before $\msgLabel{up}$ (as they are enqueued in this order). 
This also implies $\Gamma$ is deadlock-free and live: a participant first asynchronously sends all its $\msgLabel{ld}$ messages (phase 1), then awaits to receive $\msgLabel{ld}$ messages from queues of all other participants (phase 2), and only then awaits to receive $\msgLabel{up}$ messages (phase 3). Hence, all $n(n-1)$ of $\msgLabel{ld}$ and also of $\msgLabel{up}$ messages are sent/received, and a reduction path always ends with a terminated typing environment. 
\end{example}

\begin{remark}[On the decidability of typing environment properties]
\label{remark:decidability}
    Under the ``bottom-up approach'' to session typing, typing environment
    properties such as deadlock freedom and liveness are generally undecidable,
    since two session types with unbounded message queues are sufficient to
    encode a Turing machine~\cite[Theorem~2.5]{DBLP:journals/corr/abs-1211-2609}.
    Still, many practical protocols have bounded buffer sizes --- and in
    particular, the buffer sizes for the FL protocols in \cite{ProkicGKPPK23}
    are given in the same paper. Therefore,
    \Cref{ex:centralised-types,ex:decentralised-types} yield finite-state typing
    environment whose behavioural properties are decidable and easily verified,
    e.g., via model checking.
\end{remark}

\emph{Properties of our typing system.} 
We now illustrate and prove the properties of our typing system in \Cref{def:type-system}.
We aim to prove that if a session $\M$ is typed with safe/deadlock-free/live typing environment $\Gamma$, then so is $\M$.  
Along the way, we illustrate the subtle interplay between our deadlock freedom, liveness, and safety properties (\Cref{def:env-path-properties}) and subtyping (\Cref{def:subtyping}). 

First, we introduce the subtyping-related \Cref{lem:supertype-session} below, which holds directly by rule \rulename{t-sub} and is important for proving Subject Reduction later on (\Cref{thm:SR}). Based on this, in \Cref{ex:dfree-and-live-not-preserved-subtyping} we show why our definitions of deadlock-free and live typing environments (\Cref{def:env-path-properties}) also require safety. 
 
\begin{lemma}
\label{lem:supertype-session}
If $\Gamma \vdash \N$ and $\Gamma\subt \Gamma'$,  then $\Gamma' \vdash \N$.
\end{lemma}

\begin{example}
\label{ex:dfree-and-live-not-preserved-subtyping}
Consider the safe but not deadlock-free nor live session $\N$ from \Cref{ex:safe-but-not-live-session}, 
and the typing environments $\Gamma$ and $\Gamma'$ from \Cref{ex:typing-environment-properties}. 
It is straightforward to show that $\Gamma\vdash \N$, by \Cref{def:type-system}.
Observe that, as shown in \Cref{ex:subtype}, we have $\Gamma \subt \Gamma'$;
therefore, by \Cref{lem:supertype-session} we also have $\Gamma' \vdash \N$. 
As noted in \Cref{ex:typing-environment-properties}, the typing environment $\Gamma'$ is not safe.
Now, suppose that in \Cref{def:env-path-properties}, deadlock-freedom and liveness of typing environments were defined without requiring safety (as for deadlock-freedom and liveness of sessions in \Cref{def:session-properties}); also notice that $\Gamma'$ has a single path that is deadlock-free and live, but not safe.
Therefore, this change in \Cref{def:env-path-properties} would cause $\Gamma'$, an unsafe but deadlock-free and live typing environment, to type $\N$, a session that is not deadlock-free nor live. This would hamper our goal of showing that if a session is typed by a deadlock-free/live typing environment, then the session is deadlock free/live.
\end{example}

Next we show our main theoretical results: subject reduction (a reduction of a typed session can be followed by its safe/deadlock-free/live typing environment); session fidelity (if the typing environemnt reduces so can the session); and type safety, deadlock-freedom, and liveness (if typing environment is safe/deadlock-free/live, then so is the session). \emph{(Proofs are in \Cref{app:type-system}.)}
\begin{restatable}[Subject Reduction]{theorem}{ThmSubjectReduction}
  \label{thm:SR}
  Assume\; $\Gamma \vdash \N$ %
  \;with $\Gamma$ safe/deadlock-free/live %
  and\; $\N \red \N'$. %
  \;Then, there is safe/deadlock-free/live type environment $\Gamma'$ %
  such that\; $\Gamma \reds \Gamma'$ \;and\; $\Gamma'\vdash\N'$.
\end{restatable}
\vspace{-3mm}
\begin{restatable}[Type Safety]{theorem}{ThmTypeSafety}
  \label{thm:type-safety}
  If\; $\Gamma \vdash \N$ %
  \;and $\Gamma$ is safe, then 
  $\N$ is safe.
\end{restatable}
\vspace{-3mm}
\begin{restatable}[Session Fidelity]{theorem}{ThmSessionFidelity}
\label{thm:session-fidelity}
Let \;$\Gamma\vdash \N$. 
If \;$\Gamma\red$, then $\exists \Gamma', \N'$ such that \;$\Gamma\red\Gamma'$ and \;$\N\redPlus \N'$\; and \;$\Gamma'\vdash \N'$.
\end{restatable}
\vspace{-3mm}
\begin{restatable}[Deadlock freedom]{theorem}{ThmSessionDeadlockFreedom}
\label{thm:deadlock-freedom}
If $\Gamma \vdash \N$ and $\Gamma$ is deadlock-free, then $\N$ is deadlock-free.
\end{restatable}
\vspace{-3mm}
\begin{restatable}[Liveness]{theorem}{ThmLiveness}
  \label{thm:liveness}
If $\Gamma\vdash\N$ and $\Gamma$ is live, then $\N$ is live.
\end{restatable}

\section{Related and Future Work}\label{sec:related}

The original asynchronous multiparty session types \cite{HYC16} is ``top-down,'' in the sense that it begins by specifying a \emph{global type}, i.e., a choreographic formalisation of all the communications expected between the \emph{participants} in the protocol; then, the global type is \emph{projected} into a set of (local) session types, which can be used to type-check processes.  
There are many extensions to this ``top-down'' approach.  
\cite{HuY17,GheriLSTY22} extend multiparty sessions to support protocols with optional and dynamic participants, allowing sender-driven choices directed toward multiple participants. 
This is used by other work \cite{SimicPDSM21,SimicDSP24,SimicDSP25} to express communication protocols %
utilized in distributed cloud and edge computing. 
The models in \cite{MajumdarMSZ21,Stutz23,LiSWZ23,DBLP:journals/corr/abs-2411-05722} generalize the notion of projection in asynchronous multiparty session types. Building on the global types of \cite{HuY17}, they allow local types with internal and external choices directed at different participants, as we do in this paper.  Nevertheless, 
these approaches are constrained by their reliance on global types and a projection function: the global type couples both sending and receiving in a single construct and enforces a protocol structure where a single participant drives the interaction. This results in at least one projected local type beginning with a receive action --- unlike the local types for our decentralised FL protocol implementation (see \Cref{ex:decentralised-types}).

An alternative ``bottom-up'' approach to session typing \cite{ScalasY19} removes the requirement of global types:
instead, local session types are specified directly, and the properties of their composition (e.g. deadlock freedom and liveness) are checked and transferred to well-typed processes. The approach in \cite{PetersY24} extends the ``bottom-up'' synchronous multiparty session model of \cite{ScalasY19}, not only by introducing more flexible choices but also by supporting mixed choices --- i.e., choices that combine inputs and outputs. In contrast to \cite{PetersY24}, which assumes synchronous communication, our theory supports asynchronous communication, which is essential for our motivating scenarios and running examples; moreover, \cite{PetersY24} does not prove process liveness nor session fidelity.

The recent work \cite{StutzD25} presents an automata-based approach for checking multiparty protocols. Unlike ours, \cite{StutzD25} presents a top-down approach to session typing, providing a soundness and completeness result on its projection; moreover, \cite{StutzD25} shows that its typing system is also usable in ``bottom-up'' fashion, like ours. Their global types are specified as Protocol State Machines (PSMs), featuring decoupled send-receive operations, allowing for a wider range of protocols to be specified. Their local type specifications are specified as Communicating State Machines (CSMs), whereas ours are the extended session types. Their sessions are represented, as ours, as $\pi$-calculus processes (theirs with delegation), and the type system relates sessions to CSMs, and has the following distinguishing points w.r.t. our paper. First, we provide a coinductively defined subtyping relation with a standard subsumption rule \rulename{t-sub}, while \cite{StutzD25} has no subtyping relation: specifically, it embeds output subtyping (i.e., the selection of one among multiple possible outputs) within the typing rules, but does not support input subtyping. \cite{StutzD25} explicitly notes that ``there are subtleties for subtyping as one cannot simply add receives''. This makes the interplay of safety, deadlock-freedom/liveness, and subtyping we pointed out in our paper not reproducible in their setting. 
Second, \cite{StutzD25} proves a safety property (no label mismatches and terminated sessions not leaving orphan messages) and the progress (deadlock-freedom) property (if the type of the session is not final, the process can take a step). Their progress property is proven for the processes containing only one session (like ours) — but their results do not include a proof of liveness, unlike ours. 
Finally, \cite{StutzD25} specifies ``...unsafe communication: a process is stuck because all the queues it is waiting to receive from are not empty, but the labels of the first messages do not match any of the cases the process is expecting'' — while in our case it is sufficient for one queue to have an unmatching message. This is a reason why \cite{StutzD25} does not allow for subtyping, unlike our paper.

The correctness of the decentralized FL algorithm has been formally verified for deadlock-freedom and termination in \cite{ProkicGKPPK23,DjukicPPGPP25} by  using the Communicating Sequential Processes calculus (CSP) and the Process Analysis Toolkit (PAT) model checker. 
Our asynchronous multiparty session typing model abstracts protocol behavior at the type level, paving the way for more scalable and efficient techniques --- not only for model checking, but also for broader verification and analysis applications.

\emph{Conclusions and future work.} 
We present the first ``bottom-up'' asynchronous multiparty session typing theory 
that supports internal and external choices targeting multiple distinct participants. We introduce a process and type calculus, which we relate through a type system. We formally prove safety, deadlock-freedom, liveness, and session fidelity, highlighting interesting dependencies between these properties in the presence of a subtyping relation. Finally, we demonstrate how our model can represent a wide range of communication protocols, including those used in asynchronous decentralized federated learning. 

For future work, we identify two main research directions. The first focuses on the fundamental properties of our approach, including: verifying typing environment properties via model checking, in the style of \cite{ScalasY19}; 
investigating decidable approximations of deadlock freedom and liveness of typing contexts (see \Cref{remark:decidability}), e.g.~with the approach of 
\cite{LangeY19}; and investigating the expressiveness of our model based on the framework in \cite{PetersY24}.
The second direction explores the application of our model as a foundation for reasoning about federated-learning-specific properties, e.g.: handling crashes of arbitrary participants \cite{BarwellHY023}; 
 supporting optional  participation \cite{HuY17,GheriLSTY22}; ensuring that participants only receive data they are able to process; statically enforcing that only model parameters---not raw data---are exchanged to preserve privacy; guaranteeing sufficiently large server buffers to receive messages from all clients; and ensuring that all clients contribute equally to the algorithm.

\vspace{-4mm}
\subsubsection{\ackname} 
This work was supported by Horizon EU project 101093006 \emph{TaRDIS},  MSTDI (Contract No. 451-03-137/2025-03/200156) and FTS-UNS through the project \emph{Scientific and Artistic Research Work of Researchers in Teaching and Associate Positions at  FTS-UNS 2025} (No. 01-50/295), COST CA20111 \emph{EuroProofNet}, and  EPSRC grants EP/T006544/2, 
EP/N027833/2,EP/T014709/2, EP/V000462/1, EP/X015955/1, EP/Y005244/1 and ARIA.

\bibliographystyle{splncs04}
\bibliography{session}

\newpage 
\appendix

\section{Additional material of \Cref{sec:intro}}
\label{app:intro}

\begin{figure}
\centering
\includegraphics[scale=0.6]{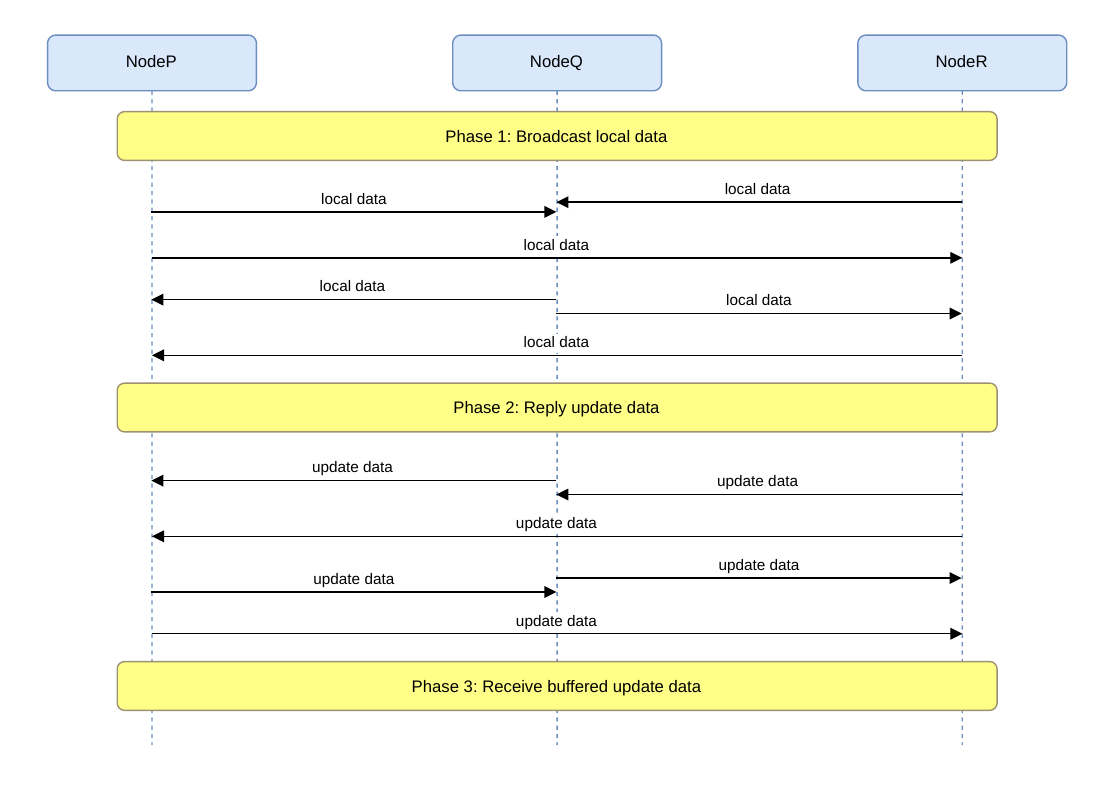}
\caption{An execution of the generic decentralised one-shot federated learning algorithm.}%
\label{fig:generic decentralised}
\end{figure}

\section{Structural (pre)congruence relation for processes}
\label{app:processes}

\begin{figure}[h]{}
\(%
\begin{array}{c}
\h_1\cdot \msg{\pq_1}{\ell_1}{\val_1} \cdot \msg{\q_2}{\ell_2}{\val_2}\cdot \h_2  \equiv  \h_1 \cdot \msg{\q_2}{\ell_2}{\val_2}\cdot \msg{\pq_1}{\ell_1}{\val_1}\cdot\h_2 \quad\text{(if $\q_1\neq \q_2$)}\\[1mm]
 \EmptyQueue \cdot\Queue \equiv  \Queue \qquad
 \Queue\cdot \EmptyQueue \equiv  \Queue \qquad
 \Queue_1\cdot (\Queue_2\cdot\Queue_3) \equiv (\Queue_1 \cdot \Queue_2)\cdot\Queue_3
\qquad
 \N_1\pc \N_2 \equiv  \N_2\pc \N_1 
 \\[1mm]
 \mu X.\PP \equivv \PP\subst{\mu X.\PP}{X}  
 \quad 
  \pa\pp\inact \pc \pa\pp\EmptyQueue \pc \N \equiv \N \quad 
 (\N_1\pc \N_2) \pc \N_3 \,\equiv\,  \N_1\pc (\N_2 \pc \N_3) 
 \\[1mm]
 \PP\equivv \Q \;\text{ and }\;\h_1 \equiv \h_2 \;\;\implies\;\; \pa\pp\PP\pc \pa\pp\h_1 \pc \N \,\equivv\, \pa\pp\Q\pc \pa\pp\h_2 \pc \N
\end{array}
\)\\
\caption{\label{tab:congruence}Structural (pre)congruence rules for queues, processes, and sessions.}
\end{figure}

The \textbf{structural precongruence} relation $\equivv$ and the \textbf{structural congruence} relation $\equiv$ (which is included in the precougruence $\equivv$)  are %
defined in a standard way in Table \ref{tab:congruence}. 
Notice that the congruence $\equiv$ for queues allows the reordering of messages with different receivers \cite{GhilezanPPSY23}:
this way, a participant $\pq_2$ can always receive the oldest queued message directed to $\pq_2$, even if that message is sent and queued after
another message directed to $\pq_1 \neq \pq_2$.
Also notice that the precongruence $\equivv$ for processes allows 
unfolding recursive processes, while folding is not allowed \cite{UdomsrirungruangY25}.  
The reason is that folding and unfolding are type-agnostic operations (types and sorts are introduced in \Cref{sec:tsu}); through folding, we may inadvertently identify two variables to which the type system (cf. \Cref{sec:tsy}) assigns different sorts.

\section{Structural congruence relation for types}
\label{app:types}

For the session types, we define $\equiv$ as the least congruence %
such that $\mu\ty. \T \equiv \T\subst{\mu\ty. \T}{\ty}$. For the queue types, we define $\equiv$ to allow reordering messages that have different recipients (similarly to the congruence for message queues in \Cref{tab:congruence}). More precisely, $\equiv$ is defined as the least congruence satisfying:

\smallskip%
\centerline{
$
\begin{array}{c}
\tqueue \cdot \temptyqueue \;\equiv\; \temptyqueue \cdot \tqueue \;\equiv\; \tqueue %
  \qquad\qquad
   \tqueue_1 \cdot (\tqueue_2 \cdot \tqueue_3) \;\equiv\; (\tqueue_1 \cdot \tqueue_2) \cdot \tqueue_3\\[1mm] 
{\tqueue \cdot \tout{\pp_1}{\ell_1}{\ST_1} \cdot \tout{\pp_2}{\ell_2}{\ST_2} \cdot \tqueue' \;\;\equiv\;\; \tqueue\cdot \tout{\pp_2}{\ell_2}{\ST_2} \cdot \tout{\pp_1}{\ell_1}{\ST_1} \cdot \tqueue'}
\quad\text{(if $\pp_1 \neq \pp_2$)}
\end{array}
$}%
\smallskip

\noindent%
Pairs of queue/session types are related via   structural congruence\; $(\tqueue_1,\T_1) \equiv (\tqueue_2, \T_2)$ \;iff\; $\tqueue_1 \equiv \tqueue_2$ and $\T_1 \equiv \T_2$. By extension, we also define congruence for typing environments (allowing additional entries being pairs of types of empty queues and terminated sessions):\; %
  $\Gamma\equiv\Gamma'$ 
  \;iff\; $\forall\pp\!\in\! \dom{\Gamma} \cap \dom{\Gamma'}{:}$ $\Gamma(\pp)\!\equiv\Gamma'(\pp)$  
and $\forall\pp\!\in\! \dom{\Gamma} \setminus \dom{\Gamma'}, \pq\!\in\!\dom{\Gamma'} \setminus \dom{\Gamma}{:}$ $\Gamma(\pp) \!\equiv\! (\temptyqueue,\tend) \!\equiv\! \Gamma'(\pq)$. %

\section{Proofs of \Cref{sec:tsu}}
\label{app:subtyping}

\begin{proposition}
  \label{lem:congruence-preserves-safety}%
  If\, $\Gamma$ is safe/deadlock-free/live and $\Gamma \equiv \Gamma'$, %
  then $\Gamma'$ is safe/deadlock-free/live.
\end{proposition}
\begin{proof}
Since the typing environment reductions are closed under structural congruence, paths $(\Gamma_i)_{i \in I}$  starting with $\Gamma_0=\Gamma$ are exactly those starting with $\Gamma_0=\Gamma'$. 
\end{proof}

\PropRedPresevesProperties*

\begin{proof}
Consider all paths $(\Gamma_i)_{i \in I}$ starting with $\Gamma_0=\Gamma$ and $\Gamma_1=\Gamma'$. If any of the properties of \Cref{def:env-path-properties} holds for $i\geq 0$ (i.e. starting with $\Gamma$) for any of considered paths, then it also must hold for $i \geq 1$ (i.e. paths starting with $\Gamma'$). Hence, if $\Gamma$ is safe/deadlock-free/live and $\Gamma \red \Gamma'$, %
then $\Gamma'$ is safe/deadlock-free/live.
\end{proof}

\begin{lemma}
  \label{lem:subtyping-shape-inp-out}%
  Let $\T'\subt\T$.
  \begin{enumerate}
  \item $\T' \equiv \texternal_{i\in I\cup J}\tin{\pq_i}{\ell_i}{\ST_i}.{\T_i'}$\; iff \; $\T \equiv \texternal_{i\in I}\tin{\pq_i}{\ell_i}{\ST_i}.{\T_i}$, where for all $i\in I$ $\T_i'\subt \T_i$, and $\{\pq_i\}_{i \in I\cup J}=\{\pq_i\}_{i \in I}$.
  \item $\T' \equiv \tinternal_{i\in I}\tout{\pq_i}{\ell_i}{\ST_i}.{\T_i'}$ \; iff \; $\T \equiv \tinternal_{i\in I\cup J}\tout{\pq_i}{\ell_i}{\ST_i}.{\T_i}$, where for all $i\in I$ $\T_i'\subt \T_i$, and $\{\pq_i\}_{i \in I}=\{\pq_i\}_{i \in I\cup J}$.
  \end{enumerate}
\end{lemma}

\LemSubtypingAndReduction*

\begin{proof}
We distinguish two cases for $\alpha$.

\noindent
\underline{Case $\alpha=\recvLabel{\pq}{\pp_k}{\ell_{k}}$:}
In this case we have:
\begin{align}
&\Gamma' \equiv {\pp_k}{:}\,(\tout{\pq}{\ell_k}{\ST_k}{\cdot} \tqueue, \T_{\pp}'),\pq{:}\,(\tqueue_{\pq}, \texternal_{i\in I\cup J}\tin{\pp_i}{\ell_i}{\ST_i}.{\T_i'}), \Gamma_2' \\
& \Gamma'_1 \equiv {\pp_k}{:}\,(\tqueue, \T_\pp'), \pq{:}\,(\tqueue_{\pq}, \T_k'),\Gamma_2'
\end{align}
where $\exists i\!\in\! I\cup J: (\pp_i, \ell_i, \ST_i)=(\pp_k, \ell_k, \ST_k)$. %
Since $\Gamma' \subt \Gamma$ we have 
$\forall\pp\!\in\! \dom{\Gamma}= \dom{\Gamma'}{:}$
  $\Gamma'(\pp)\subt\Gamma(\pp)$. 
Hence, by \Cref{lem:subtyping-shape-inp-out} we obtain
\begin{align}
&\Gamma \equiv {\pp_k}{:}\,(\tout{\pq}{\ell_k}{\ST_k}{\cdot} \tqueue, \T_{\pp}),\pq{:}\,(\tqueue_{\pq}, \texternal_{i\in I}\tin{\pp_i}{\ell_i}{\ST_i}.{\T_i}), \Gamma_2 
\end{align}
where $\T_\pp'\subt \T_\pp$, $\forall i\in I: \T_i'\subt \T_i$, $\{\pp_i\}_{i\in I \cup J}=\{\pp_i\}_{i\in I}$, and $\Gamma_2' \subt \Gamma_2$. %
Since $\Gamma$ is safe, we have $\exists i\!\in\! I: (\pp_i, \ell_i, \ST_i)=(\pp_k, \ell_k, \ST_k)$, %
and hence $\Gamma \redLabel{\;\alpha\;} \Gamma_1$ %
and $\Gamma_1'\subt \Gamma_1$ for 
\begin{align}
& \Gamma_1 \equiv {\pp_k}{:}\,(\tqueue, \T_\pp), \pq{:}\,(\tqueue_{\pq}, \T_k),\Gamma_2
\end{align}

\noindent
\underline{Case $\alpha=\sendLabel{\pp}{\pq_k}{\ell_{k}}$:}
In this case we have: 
\begin{align}
& \Gamma' \equiv \pp:(\tqueue, \tinternal_{i\in I}\tout{\pq_i}{\ell_i}{\ST_i}.{\T_i'}),\Gamma_2'\\
& \Gamma_1' \equiv \pp:(\tqueue{\cdot} \tout{\pq_k}{\ell_k}{\ST_k}\phantom{},%
   \T_k'),\Gamma_2'
\end{align}
where $k\in I$. Since $\Gamma' \subt \Gamma$, by \Cref{lem:subtyping-shape-inp-out} and we have 
\begin{align}
& \Gamma \equiv \pp:(\tqueue, \tinternal_{i\in I\cup J}\tout{\pq_i}{\ell_i}{\ST_i}.{\T_i}),\Gamma_2
\end{align}
where $\forall i\in I: \T_i'\subt \T_i$, $\{\pq_i\}_{i\in I}=\{\pq_i\}_{i\in I \cup J}$, and $\Gamma_2' \subt \Gamma_2$. %
Hence, 
$\Gamma \redLabel{\;\alpha\;} \Gamma_1$ %
and $\Gamma_1'\subt \Gamma_1$ for 
\begin{align}
& \Gamma_1 \equiv \pp:(\tqueue{\cdot} \tout{\pq_k}{\ell_k}{\ST_k}\phantom{},%
   \T_k),\Gamma_2
\end{align}
\end{proof}

\begin{lemma}
\label{lem:supertyping-and-reduction}
Let $\Gamma'\subt \Gamma$. 
\begin{enumerate}
\item If $\Gamma \redRecv{\pq}{\pp_k}{\ell_{k}}  \Gamma_1$, then $\Gamma'\redRecv{\pq}{\pp_k}{\ell_{k}} \Gamma'_1$ and $\Gamma'_1 \subt \Gamma_1$.
\item If $\Gamma\redSend{\pq}{\pp_k}{\ell_{k}} \Gamma_k$, for $k\in I \cup J$, then $\Gamma'\redSend{\pq}{\pp_i}{\ell_{i}} \Gamma'_i$  and $\Gamma'_i \subt \Gamma_i$, for $i\in I$, where $\{\pp_k\}_{k\in I\cup J}=\{\pp_i\}_{i\in I}$.
\end{enumerate}
\end{lemma}

\begin{proof}
The proof is analogous to the proof of \Cref{lem:subtyping-and-reduction}.
\end{proof}

\begin{lemma}
\label{lem:subtyping-preserves-initial-safe}
  If $\Gamma$ is safe and $\Gamma'\subt \Gamma$, then $\Gamma'$ satisfies clause \ref{item:type:path:safe} of \Cref{def:env-path-properties}.
\end{lemma}

\begin{proof}
Let $\Gamma'(\pp)\equiv (\tqueue_\pp, 
\tinternal_{j\in J\cup K}\tin{\pq_j}{\ell_j}{\ST_j}.{\T_j'})$ and $\Gamma'(\pq_k)\equiv (\tout{\pp}{\ell_k}{\ST_k}{\cdot} \tqueue_\pq, 
\T_{\pq}')$, %
   with $\pq_k\in\{\pq_j\}_{j\in J\cup K}$. %
We need to prove that then  
   $\exists j\in J\cup K: (\pq_j, \ell_j, \ST_j)=(\pq_k, \ell_k, \ST_k)$. %
If $\Gamma'(\pp)=(\tqueue_\pp, 
\T_\pp')$, then we have 
\begin{align}
\label{eq:sub-preserves-initial-safe-1}
\T_\pp' \equiv \tinternal_{j\in J\cup K}\tin{\pq_j}{\ell_j}{\ST_j}.{\T_j'}
\end{align}
Let $\Gamma(\pp)=(\tqueue_\pp, \T_\pp)$ and $\Gamma(\pq_k)=(\tout{\pp}{\ell_k}{\ST_k}{\cdot} \tqueue_\pq, \T_\pq)$. Since $\Gamma' \subt \Gamma$, we have 
$\T_\pp'\subt\T_\pp$, and $\T_\pq'\subt\T_\pq$. 
By \Cref{lem:subtyping-shape-inp-out} and \eqref{eq:sub-preserves-initial-safe-1}, we have 
\begin{align}
\label{eq:sub-preserves-initial-safe-2}
\T_\pp \equiv \tinternal_{j\in J}\tin{\pq_j}{\ell_j}{\ST_j}.{\T_j}  \quad \forall j\in J \quad \T_j'\subt\T_j \quad \text{and} \quad \{\pq_j\}_{j \in J\cup K}=\{\pq_j\}_{j \in J}
\end{align}
Now we have 
\begin{align}
\label{eq:sub-preserves-initial-safe-3}
\Gamma(\pp)\equiv (\tqueue_\pp, \tinternal_{j\in J}\tin{\pq_j}{\ell_j}{\ST_j}.{\T_j}) \;\; \text{and} \;\; \Gamma(\pq_k)\equiv (\tout{\pp}{\ell_k}{\ST_k}{\cdot} \tqueue_\pq, \T_{\pq})
\end{align}
 with $\pq_k\in\{\pq_j\}_{j \in J\cup K}=\{\pq_j\}_{j \in J}$. %
 Since $\Gamma$ is safe, we have 
 $\exists j\in J\,(\subseteq J\cup K): (\pq_j, \ell_j, \ST_j)=(\pq_k, \ell_k, \ST_k)$. %
\end{proof}

\lemSubtypingPreservesSafety*

\begin{proof} 
Let us first assume $\Gamma$ is safe, $\Gamma' \subt \Gamma$ and that $\Gamma'$ is safe. %
Now let $(\Gamma'_i)_{i\in I}$ be a  path starting with $\Gamma'$, i.e., $\Gamma'=\Gamma_0'$ and 
\begin{align}
\label{eq:sub:safe:1}
 \Gamma_0'\redLabel{\;\alpha_1\;}\Gamma_1'\redLabel{\;\alpha_2\;}\ldots
 \qquad \text{or}\qquad
 \Gamma_0'\redLabel{\;\alpha_1\;}\Gamma_1'\redLabel{\;\alpha_2\;}\ldots 
\redLabel{\;\alpha_{n}\;}\Gamma'_n
\end{align}
Let $\Gamma=\Gamma_0$. %
Since $\Gamma_0'\subt \Gamma_0$, by application of \Cref{lem:subtyping-and-reduction} we have %
$\exists \Gamma_1$ such that 
$\Gamma_1'\subt \Gamma_1$ and 
$\Gamma_0 \red \Gamma_1$. %
Since $\Gamma=\Gamma_0$ is safe, by \Cref{lem:move-preserves-safety}, $\Gamma_1$ is also safe. %
By consecutive application of the above arguments, we obtain that for all $i\in I$ there is 
$\Gamma_i$ such that 
$\Gamma_i'\subt \Gamma_i$, 
$\Gamma_i$ is safe, and 
\begin{align}
\label{eq:sub:safe:2}
 \Gamma_0\redLabel{\;\alpha_1\;}\Gamma_1\redLabel{\;\alpha_2\;}\ldots
 \qquad \text{or}\qquad
 \Gamma_0\redLabel{\;\alpha_1\;}\Gamma_1\redLabel{\;\alpha_2\;}\ldots 
\redLabel{\;\alpha_{n}\;}\Gamma_n
\end{align}
Now we can prove the main statement. 

\begin{enumerate}
\item \underline{Safety:}
Since $\Gamma_i$ is safe and $\Gamma'_i\subt \Gamma_i$, by \Cref{lem:subtyping-preserves-initial-safe}, %
$\Gamma_i'$ satisfies the clause \ref{item:type:path:safe} of \Cref{def:env-path-properties}, 
for all $i\in I$. %
Hence, by \Cref{def:env-path-properties}, path $(\Gamma'_i)_{i\in I}$ is safe. %
This implies $\Gamma'$ is safe.
\item \underline{Deadlock-freedom:} 
Assume $\Gamma$ is deadlock-free (which subsumes safe). 
We need to prove that if $I=\{0,1,\ldots,n\}$ and $\Gamma_n'\nred$, %
then $\fend(\Gamma_n')$. 
Assume $I=\{0,1,\ldots,n\}$ and $\Gamma_n'\nred$. 
Then, since $\Gamma_n'\subt\Gamma_n$, 
by \Cref{lem:supertyping-and-reduction}, we obtain $\Gamma_n\nred$.  
Since $\Gamma$ is deadlock-free, this implies $\fend(\Gamma_n)$. 
From $\Gamma_n'\subt\Gamma_n$ and $\fend(\Gamma_n)$ we can show $\fend(\Gamma_n')$. 
Hence, by \Cref{def:env-path-properties}, $\Gamma'$ is deadlock-free.
\item \underline{Liveness:} 
Assume $\Gamma$ is live (which subsumes safe). 
Assume $(\Gamma'_i)_{i\in I}$ from \eqref{eq:sub:safe:1} is fair. Let us prove that it is also live.  

Now we prove $(\Gamma_i)_{i\in I}$ from \eqref{eq:sub:safe:2} is fair. 
Assume 
$
\Gamma_i \redSend{\pp}{\pq}{\ell}$ or  $\Gamma_i \redRecv{\pp}{\pq}{\ell}
$. %
Since $\Gamma_i'\subt \Gamma_i$, by \Cref{lem:supertyping-and-reduction},
    we have 
$
\Gamma'_i \redSend{\pp}{\pq_1}{\ell_1}$, for some $\pq_1$ and $\ell_1$, or $\Gamma'_i \redRecv{\pp}{\pq}{\ell}
$. 
From $\Gamma'$ being fair we conclude 
that $\exists k, \pq', \ell'$ such that $I\ni k\geq  i$ %
and $\Gamma'_k \redSend{\pp}{\pq'}{\ell'} \Gamma'_{k+1}$ or $\Gamma'_k \redRecv{\pp}{\pq'}{\ell'} \Gamma'_{k+1}$. 
Therefore, 
$\alpha_{k+1}= \sendLabel{\pp}{\pq'}{\ell'}$ or 
$\alpha_{k+1}= \recvLabel{\pp}{\pq'}{\ell'}$ in both  
\eqref{eq:sub:safe:1} and \eqref{eq:sub:safe:2}, i.e., we obtain 
$\Gamma_k \redSend{\pp}{\pq'}{\ell'} \Gamma_{k+1}$ or $\Gamma_k \redRecv{\pp}{\pq'}{\ell'} \Gamma_{k+1}$.  
Hence, $(\Gamma_i)_{i\in I}$ from \eqref{eq:sub:safe:2} is fair.

Since $(\Gamma_i)_{i\in I}$ is fair and $\Gamma$ is live, $(\Gamma_i)_{i\in I}$ is also live. 

We now show that $(\Gamma'_i)_{i\in I}$ is live, which completes the proof. We have two cases.
\begin{itemize}
\item \ref{item:liveness:send}: $\Gamma'_i(\pp) \equiv \left(\tout\pq{\ell}\ST \cdot \tqueue \,,\, \T'\right)$. 
Since $\Gamma'_i\subt \Gamma_i$, we have 
$\Gamma_i(\pp) \equiv \left(\tout\pq{\ell}\ST \cdot \tqueue \,,\, \T\right)$, where $\T'\subt \T$. 
Since $\Gamma$ is live we have there 
    $\exists k$ such that $I\ni k\geq  i$ %
    \;and\; %
    $\Gamma_k \redRecv{\pq}{\pp}{\ell} \Gamma_{k+1}$. 
    Hence, $\alpha_{k+1}= \sendLabel{\pq}{\pp}{\ell}$ in both  
\eqref{eq:sub:safe:2} and \eqref{eq:sub:safe:1}, i.e., we obtain  $\Gamma'_k \redRecv{\pq}{\pp}{\ell} \Gamma'_{k+1}$. 
Thus, path $(\Gamma'_i)_{i\in I}$ satisfies clause \ref{item:liveness:send} of definition of live path.
\item \ref{item:liveness:recv}: $\Gamma'_i(\pp) \equiv \left(\tqueue_\pp \,,\, \texternal_{j \in J\cup L}{\tin{\pq_j}{\ell_j}{\ST_j}.{\T'_j}}\right)$. 
Since $\Gamma'_i\subt \Gamma_i$, by \Cref{lem:subtyping-shape-inp-out},  
we have  
$\Gamma_i(\pp) \equiv \left(\tqueue_\pp \,,\, \texternal_{j \in J}{\tin{\pq_j}{\ell_j}{\ST_j}.{\T_j}}\right)$, 
where $\T'_j \subt \T_j$, for $j\in J$, and 
$\{\pq_j\}_{j\in J \cup L}=\{\pq_j\}_{j\in J}$. 
From $\Gamma$ being live, we obtain there $\exists k, \pq, \ell$ such that $I\ni k\geq  i$, 
    $(\pq,\ell)\in \{(\pq_j,\ell_j)\}_{j\in J}$,%
    \;and\; %
    $\Gamma_k \redRecv{\pp}{\pq}{\ell} \Gamma_{k+1}$.
 Hence, $\alpha_{k+1}= \recvLabel{\pp}{\pq}{\ell}$ in both  
\eqref{eq:sub:safe:2} and \eqref{eq:sub:safe:1}, i.e., we obtain  there $\exists k, \pq, \ell$ such that $I\ni k\geq  i$, 
    $(\pq,\ell)\in \{(\pq_j,\ell_j)\}_{j\in J\cup L}$,%
    \;and\; %
    $\Gamma'_k \redRecv{\pp}{\pq}{\ell} \Gamma'_{k+1}$. 
Thus, path $(\Gamma'_i)_{i\in I}$ satisfies clause \ref{item:liveness:recv} of definition of live path.
\end{itemize}
\end{enumerate}
\end{proof}

\section{Proofs of Section \ref{sec:type-system}}
\label{app:type-system}

\begin{lemma}[Typing congruence]
\label{lemma:type_congruence}
\begin{enumerate}
\item
  If $\Theta \vdash \PP: \T$ and $\PP\equivv\PQ$, then $\Theta\vdash \PQ:\T.$ 
\item 
  If $\vdash \h_1:\tqueue_1$ and $\h_1\equiv \h_2$, then there is $\tqueue_2$ such that $\tqueue_1\equiv\tqueue_2$ and $ \vdash\h_2:\tqueue_2.$
\item
  If $\Gamma \vdash \N$ and $\N\equivv \N'$, then there is $\Gamma'$ such that $\Gamma \equiv\Gamma'$ and $\Gamma' \vdash \N'$.
\end{enumerate}
\end{lemma}

\begin{proof}
  The proof is by case analysis.
\end{proof}

\begin{lemma}[Substitution]
  \label{lem:substitution}%
  If $\Theta, x:\S \vdash \PP: \T$ %
  and\; $\Theta\vdash \val: \S$, %
  \;then\; %
  $\Theta \vdash \PP\sub{\val}{\x}: \T$. %
\end{lemma}
\begin{proof}
By structural induction on $\PP$.
\end{proof}

\ThmSubjectReduction*
\begin{proof}
  Assume:
  \begin{align}
    \label{lem:sr:hyp:typing}%
    &\textstyle%
    \Gamma \vdash \N%
    \quad\text{where}\quad%
    \N = \prod_{j\in J} (\pa{\pp_j}{\PP_j} \pc \pa{\pp_j}{\h_j})
    \;\;\text{(for some $J$)}
    &\text{(by hypothesis)}%
    \\%
    \label{lem:sr:hyp:gamma-safe}%
    &\text{$\Gamma$ is safe/deadlock-free/live}%
    &\text{(by hypothesis)}%
  \end{align}
 
 \noindent
    From \eqref{lem:sr:hyp:typing}, we know that there is a typing derivation
    for $\N$ starting with rule \rulename{t-sess}:
    \begin{align}
      \label{eq:sr:typing-deriv}
      \inferruleR[\rulename{t-sess}]{\Gamma=\{\pp_j{:}\,(\tqueue_j, \T_j) \;|\; j\in J\} \qquad %
      \forall j\in J \qquad \vdash \PP_j:\T_j \qquad  \vdash \h_j:\tqueue_j }{\Gamma \vdash \N}
    \end{align}

 \noindent
  We proceed by induction on the derivation of $\N \red \N'.$ 

  Base cases:
  \\
\underline{\rulename{r-send}:} We have:
    \begin{align}
      \label{eqv:proof0-0}
      &\N = \pa{\pp_n}\sum_{i\in I}\procout{\pq_i}{\ell_i}{\val_i}{\PP_i}\pc \pa{\pp_n}\h \pc  \N_1
      \\%
      \label{eqv:proof0}%
      &\N'=\pa{\pp_n}{\PP_k}\pc\pa{\pp_n}\h\cdot ({\pq_k},{\ell_k}(\val_k))\pc \N_1%
      \\%
      &\N_1 = \prod_{j\in J\setminus\{n\}} (\pa{\pp_j}{\PP_j} \pc \pa{\pp_j}{\h_j})
    \end{align}
  
   \noindent
  By inversion on the typing rules we have:
  \begin{align}
    &  \vdash \sum_{i\in I}\procout{\pq_i}{\ell_i}{\val_i}{\PP_i}:\T_{\pp_n} \qquad \text{where} \quad
    \tinternal_{i\in I}\tout{\pq_i}{\ell_i}{\S_i}.{\T_i} \subt \T_{\pp_n} \label{eqv:proof1}\\
    &  \forall i\in I \quad \vdash \val_i: \S_i \label{eqv:proof1-1}\\
    & \forall i\in I \quad  \vdash  \PP_i: \T_i \label{eqv:proof1-2}\\
    &  \vdash \h:\tqueue  \label{eqv:proof2} \\
    &  \forall j\in J\setminus\{n\} \;\;\; \vdash \PP_j:\T_j \label{eqv:proof3}\\
    &  \forall j\in J\setminus\{n\} \;\;\; \vdash \h_j:\tqueue_j \label{eqv:proof4}
    \\
        &  \Gamma_1 = \{{\pp_n}:(\tqueue, \tinternal_{i\in I}\tout{\pq_i}{\ell_i}{\S_i}.{\T_i} )\} \cup \{\pp_j:(\tqueue_j,\T_j):j\in J\setminus\{n\}\} \label{eqv:proof5} \\
        & \Gamma_1 \subt \Gamma \qquad \text{and} \quad \Gamma_1 \quad \text{is safe/deadlock-free/live by \Cref{lem:subtyping-preserves-safety}} \label{eqv:proof5-1}
  \end{align}
  
   \noindent
    Now, let:
    \begin{align}
      \label{eqv:proof11-2}%
      \Gamma_1' = \{{\pp_n}:(\tqueue \cdot  \tout{\pq_k}{\ell_k}{\S_k}, \T_k)\} \cup \{\pp_j:(\tqueue_j,\T_j):j\in J\setminus\{n\}\}
    \end{align}

 \noindent
    Then, we conclude:
    \begin{align} 
      \label{eqv:proof12}%
      & \Gamma_1 \red \Gamma_1'%
      &\text{(by \eqref{eqv:proof5}, \eqref{eqv:proof11-2}, %
        and \rulename{e-send} of Def.~\ref{def:typing-env-reductions})}%
      \\%
      &\text{$\Gamma_1'$ is safe/deadlock-free/live}%
      &\hspace{-10mm}\text{%
        (by \eqref{eqv:proof5-1}, \eqref{eqv:proof12}, and %
        Proposition~\ref{lem:move-preserves-safety})%
      }%
      \\%
      \label{eqv:proof13}
      &  \Gamma_1' \vdash \N'%
      &\text{%
        (by \eqref{eqv:proof1}--\eqref{eqv:proof4}, and
        Def.~\ref{def:type-system})%
      } \\
      \label{eqv:proof14}
      & \exists \Gamma': \quad \Gamma_1'\subt \Gamma' \quad \Gamma\red \Gamma'
      &\text{
      (by \eqref{lem:sr:hyp:typing}, \eqref{eqv:proof5-1}, \eqref{eqv:proof12}, and \Cref{lem:subtyping-and-reduction})}\\
      &\Gamma' \vdash \N' 
      &\text{
      (by \eqref{eqv:proof13}, \eqref{eqv:proof14},  and \Cref{lem:supertype-session})}\\
      & \text{$\Gamma'$ is safe/deadlock-free/live}
      & \text{
      (by \eqref{lem:sr:hyp:typing}, \eqref{eqv:proof14} and \Cref{lem:move-preserves-safety})
      }
    \end{align}
\underline{\rulename{r-rcv}:} We have:
    \begin{align}
    \label{eq:sr-receive-sess}
      &\N = \pa{\pp_n}\sum\limits_{i\in I} \procin{\pq_i}{\ell_i(\x_i)}\PP_i \pc \pa{\pp_n}\h_\pp \pc \pa{\pp_m}\PP_\pq \pc \pa{\pp_m} ({\pp_n},\ell(\val))\!\cdot \!\h \pc \N_1 \nonumber\\
     & \quad(\exists k\in I:  (\pq_k,\ell_k)=(\pp_m, \ell))%
      \\%
      \label{eq:sr-receive-cont-sess}%
      &\N' = \pa{\pp_n} \PP_k\sub{\val}{\x_k} \pc \pa{\pp_n}\h_\pp \pc  \pa{\pp_m}\PP_\pq \pc  \pa{\pp_m} \h \pc \N_1
      \\%
      &\N_1 = \prod_{j\in J\setminus\{n,m\}} (\pa{\pp_j}{\PP_j} \pc \pa{\pp_j}{\h_j})
    \end{align}

  \noindent  
  By inversion on the typing rules we have:  
  \begin{align}
    &   \vdash \sum\limits_{i\in I} \procin{\pq_i}{\ell_i(\x_i)}\PP_i:\T_{\pp_n} \qquad \text{where} \quad \texternal_{i\in I}\tin{\pq_i}{\ell_i}{\S_i}.{\T_i}\subt \T_{\pp_n} \label{eq:sr-receive-1}\\
     &   \vdash \h_\pp:\tqueue_\pp \label{eq:sr-receive-1-queue}\\
      &  \vdash \PP_\pq:\T_\pq \label{eq:sr-receive-1-1}\\
    &  \vdash ({\pp_n},\ell(\val))\!\cdot\!\h:\tout{\pp_n}{\ell}{\S} \!\cdot\! \tqueue  \label{eq:sr-receive-w} \\
    &  \forall j\in J\setminus\{n,m\}:\quad  \vdash \PP_j:\T_j \label{eq:sr-receive-3}\\
    &  \forall j\in J\setminus\{n,m\}:\quad  \vdash \h_j:\tqueue_j \label{eq:sr-receive-4}\\
    &  \Gamma_1 = \{
        {\pp_n}:(\tqueue_\pp,\texternal_{i\in I}\tin{\pq_i}{\ell_i}{\S_i}.{\T_i}),
        {\pp_m}:(\tout{\pp_n}{\ell}{\S} \!\cdot\! \tqueue ,\T_\pq)
      \}
      \cup \{\pp_j:(\tqueue_j,\T_j):j \in J\setminus\{n,m\}\}  \label{eq:sr-receive-8}\\
            & \Gamma_1 \subt \Gamma \qquad \text{and} \quad \Gamma_1 \quad \text{is safe/deadlock-free/live by \Cref{lem:subtyping-preserves-safety}}\label{eq:sr-receive-8-1}
  \end{align}

 \noindent
From \eqref{eq:sr-receive-1} and \eqref{eq:sr-receive-w} by inversion on the typing rules we have
    \begin{align}
      & \forall i \in I:\quad x_i:\S_i\vdash \PP_i: \T_i \label{eq:sr-receive-cont-t}\\
      &\vdash \h:\tqueue \;\;\text{ and } \;\;\vdash \val:\S \label{eq:sr-receive-7}
    \end{align}
 \noindent
By \eqref{eq:sr-receive-sess} 
  and $\Gamma_1$ is safe (by \eqref{eq:sr-receive-8-1}), we obtain 
   \begin{align}
   \label{eq:sr-receive-labels}
   (\pq_k,\ell_k, \ST_k)=(\pp_m, \ell, \ST)
   \end{align}
By \Cref{lem:substitution}, \eqref{eq:sr-receive-cont-t}, \eqref{eq:sr-receive-7}, and  \eqref{eq:sr-receive-labels} we have 
    \begin{align}
   \label{eq:sr-receive-subst}
   \vdash \PP_k\sub{\val}{\x_k}: \T_k
   \end{align}
    Now let:
    \begin{align}
      \label{eq:sr-receive-11}%
      &\Gamma_1' = \{
        {\pp_n}:(\tqueue_\pp,{\T_k}),
        {\pp_m}:( \tqueue ,\T_\pq)
      \}
      \cup \{\pp_j:(\tqueue_j,\T_j):j \in J\setminus\{n,m\}\} 
    \end{align}

 \noindent
    And we conclude: 
    \begin{align}
      \label{eq:sr-receive-12}%
      & \Gamma_1 \red \Gamma_1'%
      &\text{(by \eqref{eq:sr-receive-8}, \eqref{eq:sr-receive-labels}, %
        \eqref{eq:sr-receive-11}, %
        and rule \rulename{e-rcv} of Def.~\ref{def:typing-env-reductions})}
      \\%
      \label{eq:sr-receive-13}
      & \text{$\Gamma_1'$ is safe/deadlock-free/live}
      &\text{(by \eqref{eq:sr-receive-8-1}, \eqref{eq:sr-receive-12}, and \Cref{lem:move-preserves-safety})}\\
      \label{eq:sr-receive-14}
      &\Gamma_1' \vdash \N' 
      &\text{(by \eqref{eq:sr-receive-1}-\eqref{eq:sr-receive-4}, \eqref{eq:sr-receive-subst}, and \Cref{def:type-system}) }\\
      \label{eq:sr-receive-15}
      & \exists \Gamma': \quad \Gamma_1'\subt \Gamma' \quad \Gamma \red \Gamma'
      & \text{(by \eqref{eq:sr-receive-8-1}, \eqref{eq:sr-receive-12}, and \Cref{lem:subtyping-and-reduction})}\\
      &  \Gamma' \vdash \N'%
      &\text{%
        (by \eqref{eq:sr-receive-14}, \eqref{eq:sr-receive-15},  and \Cref{lem:supertype-session}))%
      }\\
      &\text{$\Gamma'$ is safe/deadlock-free/live}%
      &\hspace{-20mm}\text{%
        (by \eqref{lem:sr:hyp:gamma-safe}, \eqref{eq:sr-receive-12}, and %
        Proposition~\ref{lem:move-preserves-safety})%
      }%
    \end{align}
\\        
\underline{\rulename{r-cond-T} (\rulename{r-cond-F})}: We have:
    \begin{align}
      &\N = \pa{\pp_n}\cond\val\PP\PQ \pc \pa{\pp_n}\h \pc  \N_1
      \label{eq:sr-cond-cont-sess} \\%
      &\N' = \pa{\pp_n} \PP  \pc  \pa{\pp_n} \h \pc \N_1 \quad (\N' = \pa{\pp_n} \PQ  \pc  \pa{\pp_n} \h \pc \N_1)
      \label{eq:sr-cond-sess} \\%
      &\N_1 = \prod_{j\in J\setminus\{n\}} (\pa{\pp_j}{\PP_j} \pc \pa{\pp_j}{\h_j})
    \end{align}
  By inversion on the typing rules we have:
  \begin{align}
    &   \vdash \cond\val\PP\PQ :\T \label{eq:sr-cond-1}\\
    &  \vdash \h:\tqueue  \label{eq:sr-cond-2} \\
    &  \forall j\in J\setminus\{n\}:\quad  \vdash \PP_j:\T_j \label{eq:sr-cond-3}\\
    &  \forall j\in J\setminus\{n\}:\quad  \vdash \h_j:\tqueue_j \label{eq:sr-cond-4}\\
    &  \Gamma = \{{\pp_n}:(\tqueue,\T)\}  \cup \{\pp_j:(\tqueue_j,\T_j):j\in J\setminus\{n\}\} \label{eq:sr-cond-5}%
  \end{align}
  By inversion on the typing rules we have $\exists \T': \T'\subt\T$ such that:
    \begin{align}
      & \vdash \PP: \T'       \label{eq:sr-cond-6}\\
      & \vdash \PQ: \T'       \label{eq:sr-cond-7}\\
      & \vdash \val:\tbool       \label{eq:sr-cond-8}
    \end{align}
    But then, by rule \rulename{t-sub} we have: 
    \begin{align}
      & \vdash \PP: \T       \label{eq:sr-cond-6-1}\\
      & \vdash \PQ: \T       \label{eq:sr-cond-7-1}
    \end{align}
    Then, letting $\Gamma' = \Gamma$, we have:
    \begin{align}
      &\text{$\Gamma'$ is safe/deadlock-free/live}%
      &\hspace{-20mm}\text{%
        (by \eqref{lem:sr:hyp:gamma-safe})%
      }%
      \\%
      &  \Gamma' \vdash \N'%
      &\text{%
        (by \eqref{eq:sr-cond-sess}, \eqref{eq:sr-cond-5},  and 
        \eqref{eq:sr-cond-6-1} (or \eqref{eq:sr-cond-7-1}))%
      }%
    \end{align}
Inductive step: 
\\
 \underline{\rulename{r-struct}} Assume that $\N \red \N'$ is derived from:
    \begin{align}
      &\N \equivv \N_1 \label{lem:sr:struct1} \\
      &\N_1 \red \N_1' \\
      &\N_1' \equivv \N' \label{lem:sr:struct3}
    \end{align}
    
By \eqref{lem:sr:struct1}, \eqref{lem:sr:hyp:typing}, and \Cref{lemma:type_congruence},  there is $\Gamma_1$ such that
   \begin{align}
      & \Gamma_1 \equiv \Gamma  \label{lem:sr:struct4} \\
      &  \Gamma_1 \vdash \N_1  \label{lem:sr:struct5} 
    \end{align}
  
 \noindent 
 Since $\Gamma$ is safe/deadlock-free/live and $ \Gamma \equiv \Gamma_1$, by \Cref{lem:congruence-preserves-safety}, $\Gamma_1$ is also safe/deadlock-free/live.
  Since $\N_1 \red \N_1'$, by induction hypothesis there is a safe/deadlock-free/live type environment %
  $\Gamma_1'$ 
  such that:
  \begin{align}
      & \Gamma_1 \red \Gamma_1' \text{ or } \Gamma_1\equiv\Gamma_1'  \label{lem:sr:struct6} \\
      & \Gamma_1'\vdash\N_1'  \label{lem:sr:struct7}  
  \end{align}

 \noindent
 Since $\Gamma_1$ is safe/deadlock-free/live, by \Cref{lem:move-preserves-safety} or \Cref{lem:congruence-preserves-safety}, and \eqref{lem:sr:struct6}, we have $\Gamma_1'$ is safe/deadlock-free/live. 
 Now, by \eqref{lem:sr:struct3}, \eqref{lem:sr:struct7}, and \Cref{lemma:type_congruence}, 
 there is a environment $\Gamma'$ such that
    \begin{align}
      & \Gamma' \equiv \Gamma'_1  \label{lem:sr:struct8-0}\\
      & \Gamma' \vdash \N' \label{lem:sr:struct8}
    \end{align}
    
\noindent
 Since $\Gamma_1'$ is safe/deadlock-free/live and $ \Gamma' \equiv \Gamma_1'$, by \Cref{lem:congruence-preserves-safety}, $\Gamma'$ is also safe/deadlock-free/live.
We now may conclude with:
    \begin{align}
      &  \Gamma \red \Gamma' \text{ or } \Gamma\equiv\Gamma' & \text{(by \eqref{lem:sr:struct4}, \eqref{lem:sr:struct6}, \eqref{lem:sr:struct8-0} and rule \rulename{e-struct} of Def.~\ref{def:typing-env-reductions})}
    \end{align}
\end{proof}

\ThmTypeSafety*
\begin{proof}
Assume $\N$ is not safe. Then, there is a session path $(\N_i)_{i\in I}$, starting with $\N_0=\N$, such that for some $i\in I$
\begin{align}
&\M_i \equivv \pa\pp\sum_{j\in J} \procin{\pq_j}{\ell_j(\x_j)}\PP_j  \pc \pa\pp\h \pc \pa\pq{\PQ} \pc \pa\pq\msg{\pp}{\ell}{\val}\cdot\h_\pq \pc \M'
\nonumber\\ 
&\text{with} \quad \pq\in\{\pq_j\}_{j\in J} 
\quad \text{and} \quad \forall j\in J\;\;(\pq, \ell)\not=(\pq_j, \ell_j)
    \label{eq:type-safety-1}
\end{align}

Since $\N=\N_0\red\ldots\red\N_i$ and $\Gamma$ is safe, 
by consecutive application of \Cref{thm:SR}, we obtain $\Gamma=\Gamma_0\reds\ldots\reds \Gamma_i$, 
where $\Gamma_i$ is safe and $\Gamma_i\vdash \N_i$. 

By \Cref{lemma:type_congruence} there is 
$\Gamma_i'$ such that $\Gamma_i\equiv\Gamma'_i$ and 
\[
\Gamma'_i \vdash \pa\pp\sum_{j\in J} \procin{\pq_j}{\ell_j(\x_j)}\PP_j  \pc \pa\pp\h \pc \pa\pq{\PQ} \pc \pa\pq\msg{\pp}{\ell}{\val}\cdot\h_\pq \pc \M'
\]
Since $\Gamma_i$ is safe, by \Cref{lem:congruence-preserves-safety}, we have that $\Gamma'_i$ is safe. 
Same as in the proof of \Cref{thm:SR} (see \eqref{eq:sr-receive-8} and \eqref{eq:sr-receive-8-1}) we may conclude there is $\Gamma''_i$, such that $\Gamma''_i\subt \Gamma'_i$ and 
\[
\Gamma''_i = \{
        {\pp}:(\tqueue_\pp,\texternal_{j\in J}\tin{\pq_j}{\ell_j}{\S_j}.{\T_j}),
        {\pq}:(\tout{\pp}{\ell}{\S} \!\cdot\! \tqueue ,\T_\pq)
      \}
      \cup \Gamma'''_i  
\]
with $\pq\in\{\pq_j\}_{j\in J}$, and $\Gamma''_i$ is safe by \Cref{lem:subtyping-preserves-safety}. 
Since $\Gamma''_i$ is safe, by \ref{item:type:path:safe}, %
   $\exists j\in J: (\pq_j, \ell_j, \ST_j)=(\pq_k, \ell_k, \ST_k)$, 
   which contradicts \eqref{eq:type-safety-1}.
\end{proof}

\ThmSessionFidelity*
\begin{proof}
We distinguish two cases for deriving $\Gamma\red$.
 
\noindent
\underline{Case $1$:} $\Gamma 
\;\redRecv{\pq}{\pp_k}{\ell_{k}}\;$. 
Then: 
\[
\Gamma\equiv {\pp_k}{:}\,(\tout{\pq}{\ell_k}{\ST_k}{\cdot} \tqueue, \T_{\pp}),\pq{:}\,(\tqueue_{\pq}, \texternal_{i\in I}\tin{\pp_i}{\ell_i}{\ST_i}.{\T_i}), \Gamma_1 
\;\redRecv{\pq}{\pp_k}{\ell_{k}}\; 
{\pp_k}{:}\,(\tqueue, \T_\pp), \pq{:}\,(\tqueue_{\pq}, \T_k),\Gamma_1= \Gamma' %
\]
where $k\in I$, and where in structural congruence (starting with $\Gamma$) we can assume folding of types is not used (since only the unfolding is necessary). By inversion on the typing rules we derive  
\[
\N \equivv \pa{\pp_k}{\PP} \pc \pa{\pp_k}\msg{\pq}{\ell_k}{\val}\cdot\h \pc\pa\pq\sum_{i\in I\cup J} \procin{\pp_i}{\ell_i(\x_i)}\PP_i  \pc \pa\pp{\h_\pq} \pc  \M_1 
\]
for some $\PP, \val, \h, \PP_i$ (for $i\in I\cup J)$, $\h_\pq,$ and $\M_1$, such that 
\begin{align*}
\vdash \PP: \T_\pp \quad 
\vdash \val: \S_k \quad
\vdash \h: \tqueue \quad
\x_i:\ST_i \vdash \PP_i : \T_i \quad
\vdash \h_q: \tqueue_q \quad 
\Gamma_1 \vdash \N_1
\end{align*}
Since $k\in I(\subseteq I\cup J)$, we can show that for   
\[
\N' = \pa{\pp_k}{\PP} \pc \pa{\pp_k}\h \pc\pa\pq\PP_k\subst{\val}{\x_k}  \pc \pa\pq{\h_\pq} \pc  \M_1 
\]
we have $\N\red \N'$ and $\Gamma'\vdash \N'$.

\noindent 
\underline{Case $2$:} 
$\Gamma  
\;\;\redSend{\pp}{\pq_k}{\ell_{k}}$. We have: 
\[
\Gamma \equiv \pp:(\tqueue, \tinternal_{i\in I\cup J}\tout{\pq_i}{\ell_i}{\ST_i}.{\T_i}),\Gamma_1 
\;\;\redSend{\pp}{\pq_k}{\ell_{k}}
\]
where $k\in I\cup J$, and where in structural congruence (starting with $\Gamma$) we can assume folding of types is not used (since only the unfolding is necessary).  
By inversion on the typing rules, we have
\[
\N \equivv \pa{\pp}\sum_{i\in I}\procout{\pq_i}{\ell_i}{\val_i}{\PP_i}\pc \pa{\pp}\h \pc  \N_1
\]
for some $\val_i, \PP_i$ (for $i\in I$), $\h,$ and $\M_1$, such that 
\begin{align*} 
\vdash \val_i: \S_i \quad
\vdash \PP_i : \T_i \quad
\vdash \h: \tqueue \quad 
\Gamma_1 \vdash \N_1
\end{align*}
Now for $k\in I$ defining 
\[
\Gamma_k' = \pp:(\tqueue{\cdot} \tout{\pq_k}{\ell_k}{\ST_k}\phantom{},
   \T_k),\Gamma_1 
   \quad \text{and}  \quad
   \N_k' = \pa{\pp}{\PP_k}\pc\pa{\pp}\h\cdot ({\pq_k},{\ell_k}(\val_k))\pc \N_1
   \] 
we obtain $\Gamma\red \Gamma'_k$,\; $\N \red \N_k'$ and $\Gamma'_k \vdash \N'_k$.
\end{proof}

\ThmSessionDeadlockFreedom*
\begin{proof}
Take a path $(\N_i)_{i\in I}$, 
starting with $\N_0=\N$, 
such that $I=\{0,1,\ldots, n\}$ and $\N_n\nred$.  
Since $\N=\N_0\red\ldots\red\N_n$ and $\Gamma$ is deadlock-free, 
by consecutive application of \Cref{thm:SR}, we obtain $\Gamma=\Gamma_0\reds\ldots\reds \Gamma_n$, 
where 
$\Gamma_i\vdash \N_i$. 
Since $\N_n\nred$, by contrapositive of  \Cref{thm:session-fidelity}, we obtain $\Gamma_n\nred$. 
Since $\Gamma$ is deadlock-free, we have $\fend(\Gamma_n)$, 
and with $\Gamma_n\vdash \N_n$ and by  inversion on the typing rules we conclude $\M_n \equivv  \pa\pp{\inact} \pc \pa\pp{\emptyqueue}$. 
This proves $\N$ is deadlock-free. 
\end{proof}

\begin{remark}\label{remark:peters-df-2}
The proof of deadlock-freedom \cite[Theorem 5.2]{PetersY24} 
uses reasoning that by 
 subject reduction \cite[Theorem 5.1]{PetersY24}
 from $\Gamma'\vdash\N'$ and $\N'\nred$ one can conclude $\Gamma'\nred$. Such conclusion can be a consequence of session fidelity, not subject reduction. 
\end{remark}

\ThmLiveness*
\begin{proof}
Assume $\N$ is not live. Then, there is a fair session path $(\N_i)_{i\in I}$, starting with $\N_0=\N$, that violate one of the clauses:
\begin{itemize}
\item \ref{item:session-liveness:if}:  There is $i\in I$ such that $\M_i \equivv \pa\pp{\cond{\val}{\PP}{\PQ}} \pc \pa\pp\h \pc \M'$, %
    and 
      $\forall k$ such that $I\ni k\geq  i$ %
   we have that 
    $\M_k \equivv \pa\pp{\PP} \pc \pa\pp\h \pc \M''$
    and %
    $\M_k \equivv \pa\pp{\PQ} \pc \pa\pp\h \pc \M''$ do not hold. 
    By \Cref{thm:SR} we obtain that $\M_i$ is also typable, which by inspection of the typing rules implies that $\val$ is a boolean ($\true$ or $\false$). Hence, we have that $\M_i$ can reduce the $\pp$'s conditional process, but along path $(\N_i)_{i\in I}$ this is never scheduled. 
    This implies path $(\N_i)_{i\in I}$ also violets clause \ref{item:session-fairness:if}, and it is not fair - which is a contradiction. 
\item \ref{item:session-liveness:out}:  There is $i\in I$ such that $\M_i \equivv \pa\pp\sum_{j\in J}{\procout{\pq_j}{\ell_j}{\val_j}{\PP_j}} \pc \pa\pp\h \pc \M'$, %
    and
      $\forall k$ such that $I\ni k\geq  i$ %
    we have that %
    $\M_k   \redSend{\pp}{\pq_j}{\ell_j} \M_{k+1}$, for any $j\in J$ does not hold. 
    By \Cref{thm:SR} we obtain that $\M_i$ is also typable, which by inspection of the typing rules implies that $\val_j$ is a value (a number or a boolean) and not a variable, for all $j\in J$. This is a consequence of the fact that we can only type closed processes (by \rulename{t-sess}). Hence, we have that $\M_i$ can reduce the $\pp$'s internal choice process, but along path $(\N_i)_{i\in I}$ this is never scheduled. 
    This implies path $(\N_i)_{i\in I}$ also violets clause \ref{item:session-fairness:send}, and it is not fair - which is a contradiction. 
\item \ref{item:session-liveness:send}: 
There is $i\in I$ such that $\M_i \equivv \pa\pp\PP \pc \pa\pp\msg{\pq}{\ell}{\val}\cdot\h \pc \M'$, and %
    we have that   
      $\forall k$ such that $I\ni k\geq  i$, %
    $\M_k \redRecv{\pq}{\pp}{\ell} \M_{k+1}$ does not hold. 
    Since $\Gamma\vdash \N$ and $\Gamma$ is live, following the proof of \Cref{thm:SR},  we can construct typing environment path 
    $(\Gamma_j)_{j\in J}$, with
    $\Gamma_0=\Gamma$, where for $i\in I$, if 
    \begin{align}
    &\N_i\redSend{\pp}{\pq}{\ell}\N_{i+1} 
    \quad \text{or} \quad
    \N_i\redRecv{\pp}{\pq}{\ell}\N_{i+1} 
    \quad \text{or} \quad
    \N_i\redIf{\pp}\N_{i+1}  \quad \text{then}
    \nonumber\\
    &\Gamma_i\redSend{\pp}{\pq}{\ell}\Gamma_{i+1} 
    \quad \text{or} \quad
    \Gamma_i\redRecv{\pp}{\pq}{\ell}\Gamma_{i+1} 
    \quad \text{or} \quad
    \Gamma_i=\Gamma_{i+1} \quad \text{respecively} 
    \label{eq:thm:liveness-1}
    \end{align}
Notice that trace of $(\Gamma_j)_{j\in J}$ is exactly the same as $(\N_i)_{i\in I}$ excluding the if reductions. 
Since $(\N_i)_{i\in I}$ is fair 
(satisfies clauses \ref{item:session-fairness:send}, \ref{item:session-fairness:recv}, 
and \ref{item:session-fairness:if}), 
we conclude $(\Gamma_j)_{j\in J}$ is also fair (satisfies clauses \ref{item:fairness:send} 
and  \ref{item:fairness:recv}). 
From $\Gamma_i\vdash \N_i$, we may conclude 
$\Gamma_i(\pp) \equiv \left(\tout\pq{\ell}\ST \cdot \tqueue \,,\, \T\right)$, 
for some $\tqueue$ and $\T$. 
But then, $\forall k$ such that $J\ni k\geq  i$ %
    we have that %
    $\Gamma_k \redRecv{\pq}{\pp}{\ell} \Gamma_{k+1}$ does not hold. 
This implies $(\Gamma_j)_{j\in J}$ is fair but not live - a contradiction with assumption that $\Gamma$ is live.
\item\ref{item:session-liveness:recv}: %
There is $i\in I$ such that $\M_i \equivv \pa\pp\sum_{s\in S\cup N} \procin{\pq_s}{\ell_s(\x_s)}\PP_s  \pc \pa\pp\h \pc \M'$, %
   and  
     $\forall k$ such that $I\ni k\geq  i$,  
    we have that %
    $\M_k \redRecv{\pp}{\pq}{\ell} \M_{k+1}$, where $(\pq, \ell)\in \{(\pq_s, \ell_s)\}_{s\in S\cup N}$, does not hold.  
As in the previous case in \eqref{eq:thm:liveness-1}, 
for fair path $(\N_i)_{i\in I}$ we may construct fair path $(\Gamma_j)_{j\in J}$. 
From $\Gamma_i\vdash\N_i$ we may conclude 
 $\Gamma_i(\pp) \equiv \left(\tqueue_\pp \,,\, \texternal_{s \in S}{\tin{\pq_s}{\ell_s}{\ST_s}.{\T_s}}\right)$, for some $\tqueue_\pp, \ST_s$, and $\T_s$, for $s\in S$. 
But then, 
$\forall k, \pq, \ell$ such that $J\ni k\geq  i$ and 
$(\pq,\ell)\in \{(\pq_s,\ell_s)\}_{s\in S}$, %
we have that  %
$\Gamma_k \redRecv{\pp}{\pq}{\ell} \Gamma_{k+1}$ does not hold. 
This implies $(\Gamma_j)_{j\in J}$ is fair but not live - a contradiction with assumption that $\Gamma$ is live.
\end{itemize}
\end{proof}

\end{document}